\documentclass[10pt,journal,compsoc]{IEEEtran}

\usepackage{graphicx}
\usepackage{subfigure}
\usepackage{amsthm}
\usepackage{amsmath}
\usepackage{amssymb}
\usepackage{multirow}
\usepackage{booktabs}
\usepackage{color}
\usepackage{algorithmic}
\usepackage{algorithm}
\usepackage[flushleft]{threeparttable}
\usepackage{url}
\usepackage{pifont} 
\usepackage{boxedminipage} 

\usepackage{stmaryrd} 
\usepackage{enumitem} 
\usepackage{mathrsfs}

\newcommand{\cmark}{\ding{51}} 
\newcommand{\xmark}{\ding{55}} 

\newtheorem{definition}{Definition}
\newtheorem{theorem}{Theorem}

\ifCLASSOPTIONcompsoc
  \usepackage[nocompress]{cite}
\else
  \usepackage{cite}
\fi

\ifCLASSINFOpdf
 
\else
 
\fi

\begin{document}

\title{When Crowdsensing Meets Federated Learning: Privacy-Preserving Mobile Crowdsensing System}

\author{Bowen~Zhao,
	Ximeng~Liu*,~\IEEEmembership{Member,~IEEE,}
	and~Weineng~Chen,~\IEEEmembership{Senior Member,~IEEE,}
	\IEEEcompsocitemizethanks{\IEEEcompsocthanksitem B. Zhao is with the School of Computing and Information Systems, Singapore Management University, Singapore, Singapore. E-mail: zhaobw29@163.com\protect
		
	\IEEEcompsocthanksitem X. Liu is with the College of Mathematics and Computer Science, Fuzhou University, Fujian, China, and Cyberspace Security Research Center, Peng Cheng Laboratory, Shenzhen. E-mail: snbnix@gmail.com
	\IEEEcompsocthanksitem W.-N. Chen are with the School of Computer Science and Engineering, South China University of Technology, Guangzhou, China. E-mail: cwnraul634@aliyun.com
	\IEEEcompsocthanksitem Corresponding authors: Ximeng Liu}
	\thanks{Manuscript received $\times\times$, $\times\times\times\times$; revised $\times\times$, $\times\times$, $\times\times\times$.}}

\markboth{arXive,~Vol.~$\times\times$, No.~$\times\times$, $\times\times\times\times$}
{Shell \MakeLowercase{\textit{et al.}}: Bare Demo of IEEEtran.cls for Computer Society Journals}

\IEEEtitleabstractindextext{%
\begin{abstract}
	Mobile crowdsensing (MCS) is an emerging sensing data collection pattern with scalability, low deployment cost, and distributed characteristics. Traditional MCS systems suffer from privacy concerns and fair reward distribution. Moreover, existing privacy-preserving MCS solutions usually focus on the privacy protection of data collection rather than that of data processing. To tackle faced problems of MCS, in this paper, we integrate federated learning (FL) into MCS and propose a privacy-preserving MCS system, called \textsc{CrowdFL}. Specifically, in order to protect privacy, participants locally process sensing data via federated learning and only upload encrypted training models. Particularly, a privacy-preserving federated averaging algorithm is proposed to average encrypted training models. To reduce computation and communication overhead of restraining dropped participants, discard and retransmission strategies are designed. Besides, a privacy-preserving posted pricing incentive mechanism is designed, which tries to break the dilemma of privacy protection and data evaluation. Theoretical analysis and experimental evaluation on a practical MCS application demonstrate the proposed \textsc{CrowdFL} can effectively protect participants privacy and is feasible and efficient.
\end{abstract}

}

\maketitle

\IEEEdisplaynontitleabstractindextext

\IEEEpeerreviewmaketitle

\IEEEraisesectionheading{\section{Introduction}\label{sec:introduction}}

\IEEEPARstart{M}{obile} crowdsensing (MCS) integrates the power of both wireless network communication and crowd intelligence to reduce the deployment cost and improve the flexibility of wireless sensor network (WSN). MCS has extensive applications in environmental sensing, intelligent transport, indoor localization, behavior sensing, etc. Furthermore, MCS is considered a powerful technology in smart cities\cite{sucasas2020signature}, which can improve the ability of sensing and provide sufficient data. Compared to traditional WSN, MCS essentially outsources data collection tasks to participants.

However, as sensing data is collected by humans (i.e., participants), sensing data involves participants' private information, such as location, voiceprint, face, and activity\cite{ryoo2017privacy}. Traditional MCS systems usually rely on a sensing platform to aggregate and analyze sensing data collected by participants\cite{jin2018incentive} through machine learning\cite{capponi2019survey,liu2019boosting,liu2019floc}. Nevertheless, the sensing platform assembling data is generally untrusted\cite{wang2019towards}. For example, Facebook, as a centralized platform, collects users' data and leaks personal data to Cambridge Analytica. Moreover, participants usually communicate with other entities of MCS through the wireless network\cite{chessa2016empowering}. Unfortunately, the wireless network connection is unstable, in which participants frequently drop out. A participant consumes resources (e.g., data traffic, sensor) to participate in an MCS task, so an MCS system usually requires to provide incentives for participants\cite{zhao2020pace}. In brief, a robust MCS system needs to provide a fair incentive for each participant and protect participants' privacy.

Arguably, if we can outsource data collection tasks to participants, we can also outsource data processing tasks to participants. Federated learning (FL) proposed by McMahan \textit{et. al}\cite{mcmahan2017communication} is an effective method to outsource data processing tasks to participants, which has become a hot topic in privacy protection and machine learning communities. \texttt{FederatedAveraing} (\texttt{FedAvg}) algorithm has been proved to be feasible and effective for FL\cite{mcmahan2017communication}. Although FL has advantages in privacy protection and solving data islands\cite{yang2019federated}, local training models still leak private information\cite{liu2019boosting,nasr2019comprehensive,song2020analyzing}. Furthermore, how to encourage more participants to participate is also a challenging task\cite{yu2020fairness,yang2019federated}. Moreover, FL likewise faces the challenges of participant dropouts\cite{bonawitz2017practical}. In this paper, we research how to integrate the advantages of FL into MCS and solve facing challenges. From the perspective of existing work, it needs to tackle the following challenges.

\subsection{Related Work and Challenges}
\textbf{Privacy protection of participants.} Sensing data involves participants' private information\cite{capponi2019survey,miao2019privacy}. The online sensing platform, as a third party of data aggregation, is usually untrusted\cite{wang2019towards,capponi2019survey}. After knowing participants' private data, an untrusted sensing platform may sell participants' personal information for business recommendations and political elections analysis. To protect the privacy of participants, existing solutions usually adopt technologies including differential privacy\cite{wang2020sparse,wang2019towards}, encryption\cite{miao2017lightweight,zhao2020pace}, and so on. Admittedly, schemes\cite{miao2017lightweight,zhao2020pace} based on encryption are effective for privacy protection during data collection. However, the platform or requester needs to efficiently process collected sensing data through powerful data processing tools, such as machine learning\cite{capponi2019survey,liu2019boosting,liu2019floc}. Privacy-preserving data aggregation\cite{jiang2018data} usually fails to handle complex computations for data processing. Besides, privacy-preserving machine learning\cite{nasr2019comprehensive} and data analysis\cite{ma2019privacy} may be feasible. Unfortunately, they require data to be stored in a central server, which compromises the participant's right to be forgotten. Federated learning\cite{yang2019federated} enables data to be stored in local, however, model parameters may still disclose confidential information\cite{liu2019boosting,aono2018privacy}. In short, how to simultaneously protect participants' privacy and analyze sensing data remains to be explored.

\textbf{Solutions against participants dropout.} Due to the instability of wireless connection, participant dropouts in MCS are common\cite{xu2018practical,xu2020privacy}. For both MCS and FL, dropped participants reduce the number of sensing data or local training models, which might bring insufficient sensing and unreliable aggregation results. In general, prior solutions\cite{xu2018practical,bonawitz2017practical,liu2019boosting} employ the secret-sharing technology to restrain participant dropouts. Specifically, to restrain one participant (e.g., $p_i$) drops out, $p_i$ splits private information into multiple fragments and shares those fragments with multiple other participants. Admittedly, the secret-sharing can recover a model parameter via those fragments provided by other participants. Apparently, if $p_i$ may loss network connection with a sensing platform, it is challenging for $p_i$ to maintain multiple stable network connections and share fragments with other participants. Furthermore, when private information to be shared is huge, participants bear an unbearable communication burden. In short, it still faces a challenge on how to keep robustness against dropped participants while reducing the overhead.

\textbf{Privacy-preserving Incentive Mechanism Design.} Auction\cite{wen2015quality,wang2016quality} and posted pricing strategies\cite{zhao2020pace,qu2020posted} are often used to design a fair incentive mechanism. The incentive mechanism based on auction usually requires participants and a sensing platform to communicate multiple rounds to determine a winner. However, if an MCS system lacks a data evaluation mechanism, the winner might contribute sensing data that does not match a bid. On the contrary, the incentive mechanism based on posted pricing strategy can provide fair incentives for participants and determine participants' rewards based on participants' data\cite{han2018quality}. Unfortunately, to protect privacy, sensing data is encrypted or obfuscated. Existing solutions\cite{qu2020posted,han2018quality,yang2017designing} usually fail to protect data privacy and evaluate data quality simultaneously. Thus, it is challenging to set reasonable rewards for each participant without knowing participants sensing data. 

\subsection{Our Contributions}
To tackle the above challenges, we propose a privacy-preserving mobile crowdsensing system based on federated learning\cite{yang2019federated}, named \textsc{CrowdFL}\footnote{\textsc{CrowdFL}: mobile \underline{\textsc{Crowd}}sensing based on \underline{\textsc{F}}ederated \underline{\textsc{L}}earning}. Different from traditional MCS systems in which a sensing platform aggregates and processes participants sensing data, our proposed \textsc{CrowdFL} enables participants to store and process sensing data in participants' end-devices. Specifically, to prevent privacy leakage during data processing, we adopt federated learning to implement participants' collaborative sensing and train. To protect model privacy, participants encrypt local training models, and \texttt{FedAvg} is executed in a ciphertext field. Moreover, considering participant dropouts, discard and retransmission strategies are proposed. To set a reasonable reward for each participant and protect data privacy simultaneously, a privacy-preserving posted pricing incentive mechanism is designed. Table 1 shows a result of functional comparisons between \textsc{CrowdFL} and prior solutions. From Table 1, we can see that the proposed \textsc{CrowdFL} has advantages in privacy protection, reward distribution, and communication cost. The contributions of this paper can be summarized as follows.
\begin{table}
\centering
\label{cmptab}
\begin{threeparttable}
	\caption{Functional Comparison between \textsc{CrowdFL} and Prior Solutions}
	\centering
	\begin{tabular*}{0.49\textwidth}{lccccccc}
		\toprule
		Functions    & \cite{liu2019boosting} & \cite{liu2019floc} & \cite{zhao2020pace} & \cite{bonawitz2017practical} & \cite{xu2020privacy} & \cite{zheng2018learning} & \textsc{CrowdFL} \\\midrule
		Func 1       &         \cmark         &       \cmark       &          \cmark          &            \cmark            &           \cmark           &        \cmark        &      \cmark      \\
		Func 2       &         \cmark         &       \xmark       &          \cmark          &            \xmark            &           \cmark           &        \cmark        &      \cmark      \\
		Func 3       &         \cmark         &       \xmark       &          \xmark          &            \cmark            &           \cmark           &        \cmark        &      \cmark      \\
		Func 4       &         \xmark         &       \xmark       &          \cmark          &            \xmark            &           \xmark           &        \xmark        &      \cmark      \\
		Func 5 &         \xmark         &       N/A       &          N/A          &            \xmark            &           \xmark           &        \xmark        &  \cmark \\ \bottomrule
	\end{tabular*}
	\begin{tablenotes}
		\small
		\item \textbf{Note.} Func 1 and Func 2: privacy-preserving data aggregation and data processing, respectively; Func 3: robustness against participants dropout; Func 4: privacy-preserving reward distribution; Func 5: low communication overhead.
	\end{tablenotes}
\end{threeparttable}
\end{table}
\begin{itemize}
\item \emph{Privacy preservation during data processing}. Compared with the state-of-the-art, all sensing data collected by participants are stored in participants' end-devices in our proposed \textsc{CrowdFL}. Meanwhile, participants collaboratively train a model for a specific sensing task with the local sensing data. Moreover, to protect model privacy, we propose a privacy-preserving federated averaging algorithm (\texttt{PriFedAvg}).
\item \emph{Robustness against participant dropouts}. To keep the robustness of an MCS system against participant dropouts, we propose a discard strategy (\textsc{CrowdFL-D}) and a retransmission strategy (\textsc{CrodFL-R}), where participants are not required to execute extra computation and communication to avoid dropout.
\item \emph{Privacy-preserving reward distribution}. To address the dilemma of data quality-ware and privacy protection, we design a privacy-preserving posted pricing incentive mechanism (\texttt{PriRwd}) based on participants' encrypted model parameters. Specifically, the greater contribution of a participant to the global average model, the more rewards the participant is.
\end{itemize}

The rest of this paper is organized as follows. In Section 2, we describe the problem formulation and preliminaries. We elaborate on the workflow of \textsc{CrowdFL} in Section 3. In Section 4, we detail the proposed privacy-preserving federated averaging algorithm (\texttt{PriFedAvg}). In section 5, we describe discard and retransmission strategies to keep robustness for participant dropouts. We present our proposed privacy-preserving incentive mechanism in Section 6. In Section 7, we give privacy analysis for the proposed \textsc{CrowdFL}. The experimental evaluations are reported in Section 8. A conclusion is given in Section 9.

\section{Problem Formulation and Preliminaries}
\subsection{System Model}
\begin{figure}
\centering
\label{sysmod}
\includegraphics[width=0.88\linewidth]{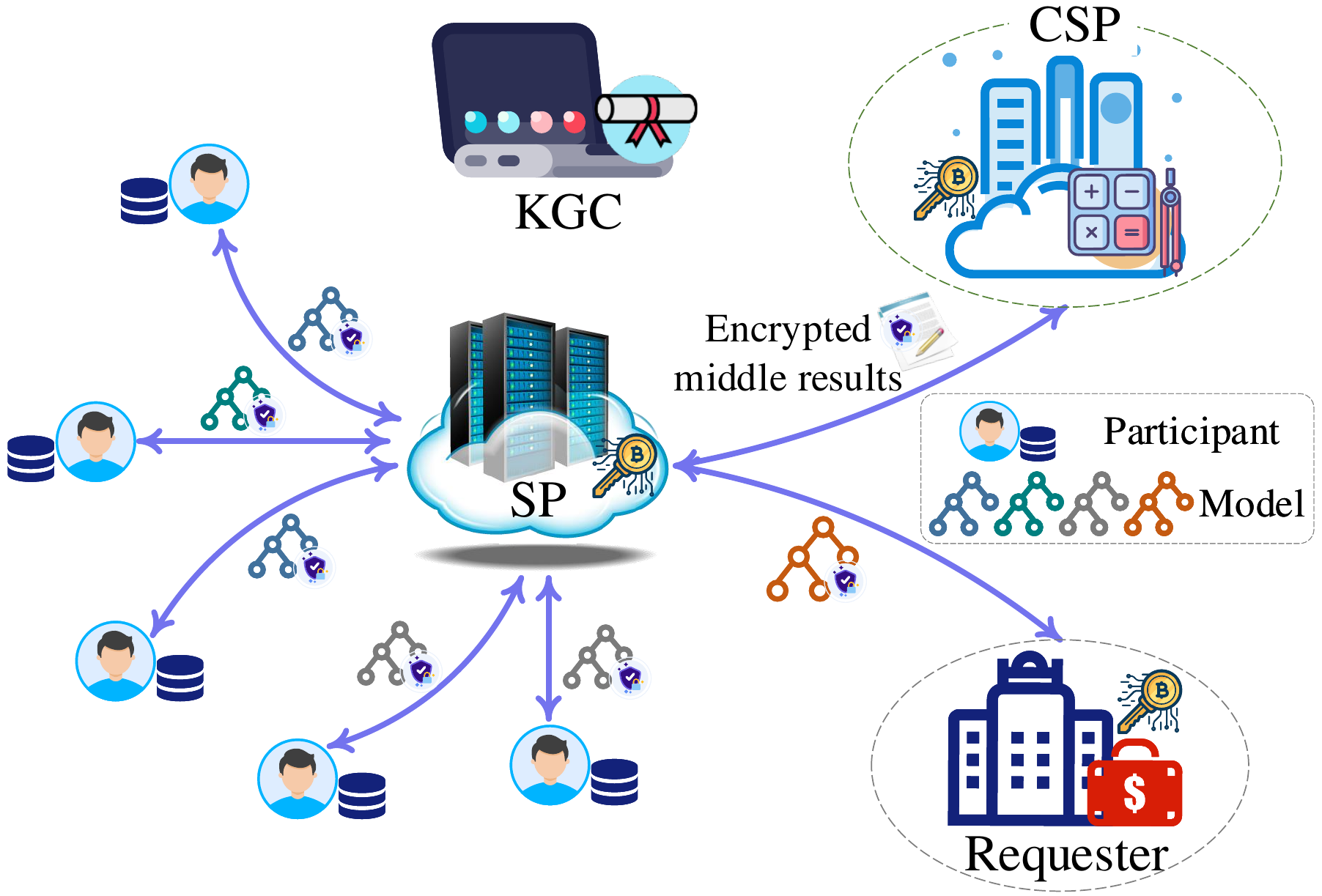}
\caption{System model of \textsc{CrowdFL}.}
\end{figure}
As depicted in Fig. 1, our proposed \textsc{CrowdFL} involves five entities, namely, a requester, participants, an online sensing platform (SP), an online computation server provider (CSP), and a key generation center (KGC).
\begin{itemize}
\item \textbf{Requester}: A requester requests participants to collect sensing data and collaboratively train a model in a federated learning manner. Besides, the requester provides rewards for participants who participate in data collection and model training. Particularly, the model aggregation is outsourced to an SP, and the requester only obtains a finally encrypted training model.
\item \textbf{Participants}: A participant collects sensing data with a mobile intelligent device and stores sensing data in her local device. After collecting sensing data, the participant trains a model with collected sensing data and submits an encrypted training model to an SP.
\item \textbf{SP}: An SP is responsible for recruiting participants and averaging participants' encrypted training models. Moreover, the SP takes charge of paying rewards for each participant.
\item \textbf{CSP}: A CSP provides computation services for the SP to assist to execute secure computation, such as \texttt{PriFedAvg} and \texttt{PriRwd}.
\item \textbf{KGC}: A KGC takes charge of distributing public/private keys to other entities.
\end{itemize}

\subsection{Threat Model}
In \textsc{CrowdFL}, entities except for the KGC are considered as \textit{honest-but-curious}\cite{liu2016efficient,liu2019boosting} who follows protocols but tries to learn other entities private information. Specifically, SP and CSP may try to learn participants' training models and a global average model. Participants might attempt to obtain other participants' training models. A requester tries to learn participants' training models and might refuse to pay rewards for participants. Moreover, we assume that SP and CSP do not collude. One possible construction is that SP and CSP belong to different organizations or companies. No collusion assumption is widely used in secure computation scenarios\cite{xu2020privacy,mohassel2017secureml,zheng2018learning,xu2019verifynet}, such as the state-of-the-art of privacy-preserving Federated learning\cite{xu2020privacy,xu2019verifynet}. To prevent the SP from obtaining participants' and the requester's private information, the requester and participants do not collude with SP. The KGC is responsible for managing keys in \textsc{CrowdFL} and is considered a trusted entity, such as a trusted third party (TA).

\subsection{Preliminaries}
\subsubsection{Federated Average Algorithm}
Federated average algorithm (\texttt{FedAvg}) is an iterative algorithm, and consists of local training and model aggregation\cite{mcmahan2017communication}. Suppose $n$ participants are indexed by $i$ (each participant $i$ owns a data set $D_i$), $B$ is the local minibatch size, $E$ is the number
of local epochs, and $\eta$ is the learning rate. In each round iteration, \texttt{FedAvg} processes as follows:
\begin{itemize}
\item[(1)] \textbf{Local training:} Participants split $D_i$ into batches of size $B$ denoted by $\mathcal{B}$. For each local epoch $i$ from 1 to $E$ and batch $b\in\mathcal{B}$, participants calculate
\begin{align}
	\omega\leftarrow \omega-\eta\nabla\ell(\omega; b),
\end{align}
where $\ell(\omega; b)$ denotes the loss of the prediction on example
$b$ made with model $\omega$. After training, participants upload the model $\omega$ to an aggregation server.
\item[(2)] \textbf{Model aggregation:} After receiving $n$ training models of participants, an aggregation server compute a global average model
\begin{align}
	\bar{\omega}\leftarrow\sum_{i=1}^n\frac{\delta_i}{\sum_{i=1}^{n}\delta_i}\omega_i,
\end{align}
where $\delta_i=|D_i|$, i.e., the size of $D_i$.
\end{itemize}

\subsubsection{Paillier Cryptosystem with Threshold Decryption (PCTD)}
PCTD is a variant of the Paillier cryptosystem\cite{paillier1999public}. The key idea of PCTD is to split the private key of the Paillier cryptosystem into two parts. Any single part cannot effectively decrypt a given ciphertext encrypted by the Paillier cryptosystem. PCTD comprises the following algorithms:

\textbf{KeyGen}: Let $\zeta$ be a security parameter and $p, q$ be two large prime number with $\zeta$ bits. Compute $N=p\cdot q$, $\lambda=(p-1)\cdot(q-1)$, and $u=\lambda^{-1}\mod N$. Choose $g=N+1$. The public key is denoted by $pk=(g, N)$, and the private key is denoted by $sk=(\lambda,u)$. 

The private key $\lambda$ is split into two parts denoted by $sk_1=\lambda_1$ and $sk_2=\lambda_2$, s.t., $\lambda_1+\lambda_2=0\mod\lambda$ and $\lambda_1+\lambda_2=1\mod N$. According
to the Chinese remainder theorem\cite{pei1996chinese}, we can calculate $\varepsilon=\lambda_1+\lambda_2=\lambda\cdot u\mod(\lambda\cdot N)$ to make $\varepsilon=0\mod\lambda$ and $\varepsilon=1\mod N$ hold simultaneously, where $\lambda_1\in[1, \lambda\cdot u]$ and $\lambda_2=\varepsilon-\lambda_1$.

\textbf{Encryption} (\texttt{Enc}): Given a message $m\in\mathbb{Z}_N$, select a random number $r\in\mathbb{Z}_N^*$, and the ciphertext of $m$ is generated as
$$\llbracket m\rrbracket=g^m\cdot r^N\mod N^2=(1+m\cdot N)\cdot r^N\mod N^2.$$

\textbf{Decryption} (\texttt{Dec}): Take a given ciphertext $\llbracket m\rrbracket$ and $sk$ as inputs, compute
$$m=L(\llbracket m\rrbracket^\lambda\mod N^2)\cdot u\mod N,$$
where $L(x)=\frac{x-1}{N}$.

\textbf{Partial Decryption} (\texttt{PDec}): Take a cihpertext $\llbracket m\rrbracket$ and partial private key $sk_i$ as inputs ($i\in\{1, 2\}$), \texttt{PDec} calculates as follows:
$$M_i=\llbracket m\rrbracket^{\lambda_i}=r^{\lambda_i\cdot N}\cdot(1+m\cdot N\cdot \lambda_i)\mod N^2.$$

\textbf{Threshold Decryption} (\texttt{TDec}): Take partial decrypted ciphertexts $\langle M_1,M_2\rangle$ as inputs, \texttt{TDec} computes
$$m=L(M_1\cdot M_2\mod N^2).$$

The only difference between the PCTD and the Paillier cryptosystem is the former sets two partial private keys $\lambda_1$ and $\lambda_2$. If the Paillier cryptosystem is semantic security\cite{paillier1999public}, the PCTD is also semantic security. Moreover, known one partial private key (e.g., $\lambda_1$), it is not feasible to calculate the others (i.e., $\lambda_2$) due to $\lambda$ being private. Thus, any adversary cannot decrypt a ciphertext by only one partial private key. The PCTD has additive homomorphism and scalar-multiplication homomorphism properties. Specifically, given $\llbracket a\rrbracket$, $\llbracket b\rrbracket$, and $c\in\mathbb{Z}_N$, $\texttt{Dec}(\llbracket a+b\rrbracket)=\texttt{Dec}(\llbracket a\rrbracket\cdot\llbracket b\rrbracket)$ and $\texttt{Dec}(\llbracket c\cdot a\rrbracket)=\texttt{Dec}(\llbracket a\rrbracket^c)$ hold.

\section{Overview of \textsc{CrowdFL}}
In this section, we first elaborate on the workflow of \textsc{CrowdFL}. Then, we
present two building blocks used in \textsc{CrowdFL}.

\subsection{Overview}
Fig. 2 shows the workflow of \textsc{CrowdFL}. The detailed processes are listed as follows. 

(1) \textit{Setup}: A requester first setups sensing data collection tasks and formulates data processing tasks into a federated learning task. Moreover, KGC generates a public/private key pair $\langle pk, sk\rangle$ and two partial decryption key pairs $\langle\lambda_c, \lambda_{sc}\rangle$ and $\langle\lambda_p, \lambda_{sp}\rangle$. KGC distributes $pk,\lambda_p$ to participants, $pk,\lambda_{sc},\lambda_{sp}$ to SP, and $pk,\lambda_c$ to CSP.

(2) \textit{Local pre-training}: Participants firstly collect sensing data and locally train a model according to sensing tasks and data processing tasks. To protect privacy, participants encrypt training models $\{\omega_1,\cdots,\omega_n\}$ with $pk$, and send $\{\llbracket \omega_1\rrbracket,\cdots, \llbracket \omega_n\rrbracket\}$ to SP.

(3) \textit{Model aggregation}: SP aggregates encrypted training models and carries out the proposed \texttt{PriAvgFed} algorithm to generate an encrypted global average model $\llbracket\bar{\omega}\rrbracket$, where the proposed \texttt{PriAvgFed} algorithm requires SP and CSP to jointly participate in a privacy-preserving manner. After that, SP calculates one partially decrypted global average model $\llbracket\bar{\omega}_s\rrbracket$ and returns $\langle\llbracket\bar{\omega}_s\rrbracket, \llbracket\bar{\omega}\rrbracket\rangle$ back to participants.

(4) \textit{Local retraining}: Participants decrypt the encrypted global average model $\llbracket\bar{\omega}\rrbracket$ to obtain another partially decrypted global average model $\bar{\omega}_p$. Participants take two partially decrypted global average models $\llbracket\bar{\omega}_s\rrbracket$ and $\llbracket\bar{\omega}_p\rrbracket$ as inputs to calculate a decrypted global average model $\bar{\omega}$. Then, participants retrain new models based on collected sensing data and $\bar{\omega}$, encrypt new training models, and send encrypted models to SP. Steps (3) and (4) are repeated until reaching maximum training times. The requester obtains the finally encrypted global average model.

(5) \textit{Reward distribution}: SP and CSP jointly execute our proposed privacy-preserving reward distribution protocol (\texttt{PriRwd}) to distribute rewards for each participant.
\begin{figure}
\label{overview}
\centering
\includegraphics[width=0.99\linewidth]{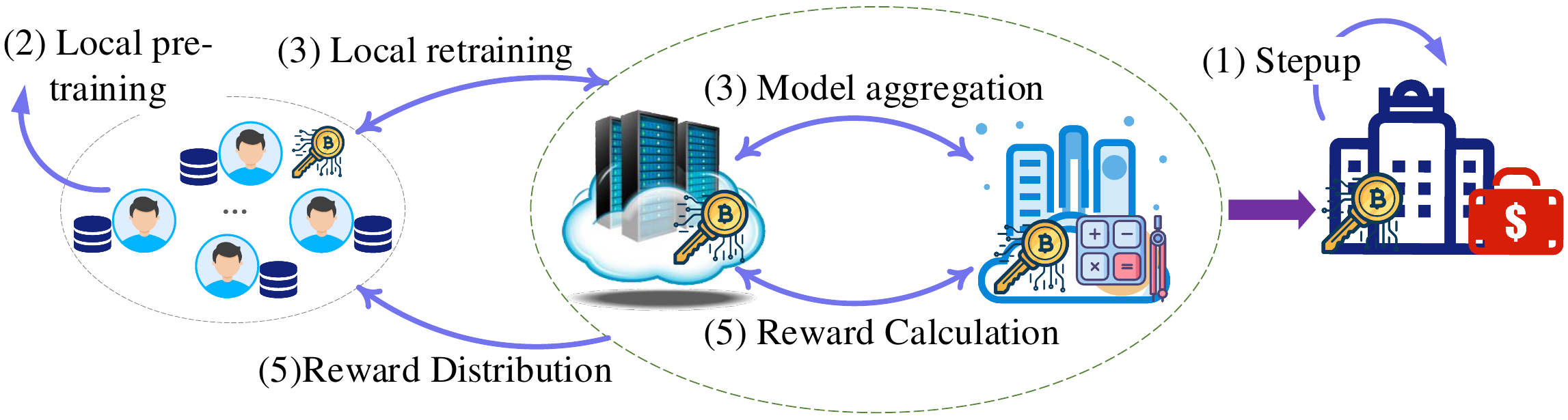}
\caption{Workflow of \textsc{CrowdFL}.}
\end{figure}

Note that if a participant drops out due to network connection loss in Step (4), the SP adopts our proposed discard or retransmission strategies to handle it. The details are given in Section 5.

\subsection{Building Blocks}
To prevent training models from leaking privacy and realize quality-aware in a privacy-preserving manner, we propose two critical building blocks secure division protocol (\texttt{SDIV}) and secure multiplication protocol (\texttt{SMUL}) to construct a privacy-preserving federated averaging algorithm (\texttt{PriFedAvg}) and privacy-preserving reward distribution mechanism (\texttt{PriRwd}).

\subsubsection{Secure Division Protocol (\texttt{SDIV})}
\begin{algorithm}[htb]  
\caption{$\texttt{SDIV}(\llbracket x\rrbracket, \llbracket y\rrbracket)\rightarrow\llbracket x\div y\rrbracket$} 
\label{sdiv} 
\begin{algorithmic}[1]
	\REQUIRE SP has $\llbracket x\rrbracket$ and $\llbracket y\rrbracket$, where $x,y\in(0, 2^\kappa)$;\\
	\STATE SP:
	\begin{itemize}[itemsep=0pt,topsep=1.25pt,itemindent=5pt]
		\item[(a).] $X\leftarrow\llbracket x\rrbracket^r\cdot\llbracket y\rrbracket^{r\alpha+e}$ and $Y\leftarrow\llbracket x\rrbracket^r$, where $\alpha$ and $e$ are random numbers with $\sigma$ bits and $\kappa$ bits, respectively, and $r$ is a random composite number in $[2^{\kappa+1}, 2^{\log_2N-\kappa-\sigma-2})$;
		\item[(b).] $X_1\leftarrow\texttt{PDec}(X, \lambda_{sc})$ and $Y_1\leftarrow\texttt{PDec}(Y, \lambda_{sc})$;
		\item[(c).] Send $\langle (X, X_1), (Y, Y_1)\rangle$ to CSP.\
	\end{itemize}
	\STATE CSP:
	\begin{itemize}[itemsep=0pt,topsep=1.25pt,itemindent=5pt]
		\item[(a).] $X_2\leftarrow\texttt{PDec}(X, \lambda_c)$ and $Y_2\leftarrow\texttt{PDec}(Y, \lambda_c)$;
		\item[(b).] $x'\leftarrow\texttt{TDec}(X_1, X_2)$ and $y'\leftarrow\texttt{TDec}(Y_1, Y_2)$;
		\item[(c).] $\llbracket x'\div y'\rrbracket\leftarrow\texttt{Enc}([x'\div y'],pk)$;
		\item[(d).] Send $\llbracket x'\div y'\rrbracket$ to SP.
	\end{itemize}
	\STATE SP:
	\begin{itemize}[itemsep=0pt,topsep=1.25pt,itemindent=5pt]
		\item[(a).] $\llbracket\alpha\rrbracket\leftarrow\texttt{Enc}(\alpha,pk)$ and $\llbracket e\div r\rrbracket\leftarrow\texttt{Enc}([e\div r],pk)$;
		\item[(b).] $\llbracket x\div y\rrbracket\leftarrow\llbracket x'\div y'\rrbracket\cdot\llbracket\alpha\rrbracket^{N-10^L}\cdot\llbracket e\div r\rrbracket^{N-1}$.
	\end{itemize}
\end{algorithmic}
\end{algorithm}
\texttt{SDIV} needs cooperation between SP and CSP. Particularly, SP and CSP do not know each other's private information. Specifically, SP has $\llbracket x\rrbracket$ and $\llbracket y\rrbracket$, where $x,y\in(0, 2^\kappa)$ and $\kappa$ is the number of bits on the upper bound of $x,y$, such as $\kappa=32$. The goal of SP is to obtain $\llbracket x\div y\rrbracket$. If $x\div y$ is decimal, we use a rounding factor $L$ to transform $x\div y$ into an integer, i.e., $[x\div y]\leftarrow \lfloor x\div y\cdot 10^L\rfloor$. For simplify, $\llbracket [x\div y]\rrbracket$ is denoted by $\llbracket x\div y\rrbracket$. Let $\sigma$ be the statistical security parameter, such as $\sigma=80$. The processes of \texttt{SDIV} are briefly described below. (1) SP obscures $x$ and $y$ by adding random noises in the ciphertexts. After that, SP partially decrypts $X,Y$ into $X_1, Y_1$ with a partial key $\lambda_{sc}$, and sends $\langle (X, X_1), (Y, Y_1)\rangle$ to CSP. (2) CSP calls \texttt{PDec} and \texttt{TDec} to obtain $x'$ and $y'$. Next, CSP calculates $x'\div y'$. It is easy to verify that $x'\div y'=x\div y+\alpha+e\div y$. CSP sends the encrypted $\llbracket x'\div y'\rrbracket$ to SP. (3) SP disposes noise $[\alpha]$ and $[e\div r]$. More details of \texttt{SDIV} are shown in Algorithm \ref{sdiv}.

According to the property of floor function, we have
\begin{equation}
[x\div y+\alpha+e\div r]=\left\{
\begin{aligned}
	&[x\div y]+[\alpha]+[e\div r]+1, \\
	&[x\div y]+[\alpha]+[e\div r].
\end{aligned}\nonumber
\right.
\end{equation}
Thus, we can learn that 
\begin{equation}
[x'\div y'-\alpha-e\div r]\!=\!\left\{
\begin{aligned}
	&[x\div y]+1,\mbox{ for $x\!\!\!\!\!\mod\!y+e\!\!\!\!\!\mod\!r\geq 1$}\\
	&[x\div y],\mbox{ for $x\!\!\!\!\!\mod\!y+e\!\!\!\!\!\mod\!r<1$}.
\end{aligned}\nonumber
\right.
\end{equation}
And since $\frac{1}{2}<\frac{2}{3}<\cdots<\frac{2^\kappa-2}{2^\kappa-1}$ and $\frac{e}{r}\leq\frac{2^\kappa-1}{2^{\kappa+1}}$, we have $$x\!\!\!\mod y+e\!\!\!\mod r<\frac{2^\kappa-2}{2^\kappa-1}+\frac{2^\kappa-1}{2^{\kappa+1}}<1.$$ Hence, we can learn that $[x'\div y'-\alpha-e\div r]=[x\div y]$ always holds.

\subsubsection{Secure Multiplication Protocol (\texttt{SMUL})}
The goal of \texttt{SMUL} is to enable SP that has $\llbracket x\rrbracket$ and $\llbracket y\rrbracket$ to obtain $\llbracket x\cdot y\rrbracket$. As the Paillier cryptosystem only supports additive homomorphism and scalar-multiplication homomorphism, \texttt{SMUL} needs the assist of CSP. To prevent CSP from learning $x$ and $y$, SP adds random noise to protect $x,y$. The steps of \texttt{SMUL} are briefly described below. (1) SP first adds random noises $r_1,r_2$ into $\llbracket x\rrbracket$ and $\llbracket y\rrbracket$ to obscure $x,y$. Then, SP partially decrypts ciphertexts added noises and sends $\langle (X, X_1), (Y, Y_1)\rangle$ to CSP. (2) CSP calls \texttt{PDec} and \texttt{TDec} to obtain $x'$ and $y'$ and calculate $x'\cdot y'$. After that, CSP encrypts $x'\cdot y'$ and sends $\llbracket x'\cdot y'\rrbracket$ to SP. It easily verifies that $x'\cdot y'=x\cdot y+r_2x+r_1y+r_1r_2$. (3) To obtain $\llbracket x\cdot y\rrbracket$ for given $\llbracket x'\cdot y'\rrbracket$, SP first calculates $\llbracket r_2\cdot x\rrbracket$, $\llbracket r_1\cdot y\rrbracket$, and $\llbracket r_1\cdot r_2\rrbracket$. Finally, SP uses the additive homomorphism and scalar-multiplication homomorphism to obtain $\llbracket x\cdot y\rrbracket$. Algorithm \ref{smul} shows more details of \texttt{SMUL}.
\begin{algorithm}[htb] 
\caption{$\texttt{SMUL}(\llbracket x\rrbracket, \llbracket y\rrbracket)\rightarrow\llbracket x\cdot y\rrbracket$} 
\label{smul} 
\begin{algorithmic}[1]
	\REQUIRE SP has $\llbracket x\rrbracket$ and $\llbracket y\rrbracket$, where $x,y\in(-2^\kappa, 2^\kappa)$;\\
	\STATE SP:
	\begin{itemize}[itemsep=0pt,topsep=1.25pt,itemindent=5pt]
		\item[(a).] $X\leftarrow\llbracket x\rrbracket\cdot\llbracket r_1\rrbracket$ and $Y\leftarrow\llbracket y\rrbracket\cdot\llbracket r_2\rrbracket$, where $r_1,r_2$ are random numbers with $\sigma$ bits;
		\item[(b).] $X_1\leftarrow\texttt{PDec}(X, \lambda_{sc})$ and $Y_1\leftarrow\texttt{PDec}(Y, \lambda_{sc})$;
		\item[(c).] Send $\langle (X, X_1), (Y, Y_1)\rangle$ to CSP.\
	\end{itemize}
	\STATE CSP:
	\begin{itemize}[itemsep=0pt,topsep=1.25pt,itemindent=5pt]
		\item[(a).] $X_2\leftarrow\texttt{PDec}(X, \lambda_c)$ and $Y_2\leftarrow\texttt{PDec}(Y, \lambda_c)$;
		\item[(b).] $x'\leftarrow\texttt{TDec}(X_1, X_2)$ and $y'\leftarrow\texttt{TDec}(Y_1, Y_2)$;
		\item[(c).] $\llbracket x'\cdot y'\rrbracket\leftarrow\texttt{Enc}(x'\cdot y',pk)$;
		\item[(d).] Send $\llbracket x'\cdot y'\rrbracket$ to SP.
	\end{itemize}
	\STATE SP:
	\begin{itemize}[itemsep=0pt,topsep=1.25pt,itemindent=5pt]
		\item[(a).] $\llbracket r_2\cdot x\rrbracket\leftarrow\llbracket x\rrbracket^{r_2}$, $\llbracket r_1\cdot y\rrbracket\leftarrow\llbracket y\rrbracket^{r_1}$, and $\llbracket r_1\cdot r_2\rrbracket\leftarrow\texttt{Enc}(r_1\cdot r_2,pk)$;
		\item[(b).] $\llbracket x\cdot y\rrbracket\leftarrow\llbracket x'\cdot y'\rrbracket\cdot\llbracket r_2\cdot x\rrbracket^{N-1}\cdot\llbracket r_1\cdot y\rrbracket^{N-1}\cdot\llbracket r_1\cdot r_2\rrbracket^{N-1}$.
	\end{itemize}
\end{algorithmic}
\end{algorithm}

\section{Privacy-preserving Federated Averaging}
In this section, we detail our proposed privacy-preserving federated averaging algorithm (\texttt{PriFedAvg}). Given encrypted training models, \texttt{PriFedAvg} requires SP and CSP to jointly calculate an encrypted global average model, where either SP or CSP cannot obtain the global average model alone. Specifically, the additive homomorphism property of PCTD is used to compute the sum of encrypted training models. Besides, the proposed \texttt{SDIV} is used to calculate the average of encrypted training models.

\subsection{Training Model Encryption}
The Paillier cryptosystem only works on the integer field, that is, a message $m$ to be encrypted should be an element of $\mathbb{Z}_N$. However, $\omega_i$ may be decimal. To tackle this problem, we introduce a rounding factor $L$ that is larger than the maximum number of decimal places of $\omega_i$. Formally, $[\omega_i]\leftarrow\lfloor\omega_i\cdot 10^L\rfloor$. For example, suppose the maximum number of decimal places is 5 and $L=6$ ($L>5$), given a decimal $a=-1.234$, $[a]=-1.234\cdot 10^6=-123400$. Particularly, if $a=1.2345678$ and $L=6$, $[a]=1234567$.

In \textsc{CrowdFL}, we assume all decimals are denoted by the fixed-point number and the decimal places is $\ell$. Given a decimal number $a$ with finite decimal places and the rounding factor $L$ ($L>\ell$), $a$ is encrypted into $\llbracket[a]\rrbracket$ (in this paper, $\llbracket[a]\rrbracket$ is abbreviated to $\llbracket a\rrbracket$). Noted that if $[a]=-5$ and $N=21$, $\llbracket[a]\rrbracket=\llbracket 16\rrbracket$. Moreover, suppose $L=6$ and $h=\texttt{Dec}(\llbracket 31415\rrbracket)$, $h$ denotes the decimal $0.031415$.

\subsection{\texttt{PriFedAvg} Design}
\begin{algorithm}[htb] 
\caption{$\texttt{PriFedAvg}(\{\langle\llbracket \delta_i\cdot \omega_i\rrbracket,\llbracket\delta_i\rrbracket\rangle\}_{i=1}^n)\rightarrow\llbracket\bar{\omega}_T\rrbracket$} 
\label{fedavg} 
\begin{algorithmic}[1]
	\REQUIRE SP has $\{\langle\llbracket \delta_i\cdot \omega_i\rrbracket,\llbracket\delta_i\rrbracket\rangle\}_{i=1}^n$, participants have $\{D_i\}_{i=1}^n$, and the requester obtains $\llbracket\bar{\omega}_T\rrbracket$;\\
	\STATE Participants train initial models and upload encrypted training models messages $\{\langle\llbracket \delta_i\cdot \omega_i\rrbracket,\llbracket\delta_i\rrbracket\rangle\}_{i=1}^n$ to SP;
	\FOR{$j=1$ to $T$}
	\STATE SP initializes $\mathcal{M}=\llbracket \delta_1\cdot \omega_1\rrbracket$ and $\mathcal{D}=\llbracket \delta_1\rrbracket$;
	\FOR{$i=2$ to $n$}
	\STATE SP calculates $\mathcal{M}\leftarrow\mathcal{M}\cdot\llbracket \delta_i\cdot \omega_i\rrbracket$ and $\mathcal{D}\leftarrow\mathcal{D}\cdot\llbracket\delta_i\rrbracket$;
	\ENDFOR
	\STATE SP and CSP jointly execute $\llbracket\bar{\omega}_j\rrbracket\leftarrow\texttt{SDIV}(\mathcal{M}, \mathcal{D})$;
	\IF{$j$ equals to $T$}
	\STATE SP sends $\llbracket\bar{\omega}_T\rrbracket$ to the requester and exits algorithm;
	\ENDIF
	\STATE SP computes $\llbracket\bar{\omega}_{j,1}\rrbracket\leftarrow\texttt{PDec}(\llbracket\bar{\omega}_j\rrbracket, \lambda_{sp})$ and sends $\langle\llbracket\bar{\omega}_j\rrbracket, \llbracket\bar{\omega}_{j,1}\rrbracket\rangle$ to participants;
	\STATE Participants calculates $\llbracket\bar{\omega}_{j,2}\rrbracket\leftarrow\texttt{PDec}(\llbracket\bar{\omega}_j\rrbracket, \lambda_p)$ and $\bar{\omega}_j\leftarrow\texttt{TDec}(\llbracket\bar{\omega}_{j,1}\rrbracket, \llbracket\bar{\omega}_{j,2}\rrbracket)$;
	\STATE Participants take $\bar{\omega}_j$ and sensing data set $\{D_i\}_{i=1}^{n}$ as inputs and retrain new models. Next, participants send encrypted training models messages $\{\langle\llbracket \delta_i\cdot \omega_i\rrbracket,\llbracket\delta_i\rrbracket\rangle\}_{i=1}^n$ to SP;
	\ENDFOR
\end{algorithmic}
\end{algorithm}
Based on the proposed \texttt{SDIV}, we can design a privacy-preserving federated average algorithm (\texttt{PriFedAvg}), which enables SP and CSP to jointly calculate an encrypted global average model. In \texttt{PriFedAvg}, SP and CSP do not have knowledge of the global average model. \texttt{PriFedAvg} is shown in Algorithm \ref{fedavg}. $\delta_i$ is the amount of a participant $i$ owning training data items, i.e., $\delta_i=|D_i|$. $T$ denotes the maximum training rounds. A brief description of Algorithm \ref{fedavg} is given below.

Participants train local models based on collected sensing data and the aggregated global average model, and encrypt and upload encrypted training models to SP (Line 1 and Line 13). To generate a global average model in a privacy-preserving manner, after aggregating participants' encrypted training models, SP first calculates the sum of encrypted training models (Lines 3-6). Next, SP and CSP jointly carry out the proposed \texttt{SDIV} to compute the global average model in a ciphertext field (Line 7). In general, each participant needs to obtain a global average model to train a new model before reaching maximum training rounds. To this end, SP partially decrypts the encrypted average model, and then the participant can compute the global average model by using \texttt{PDec} and \texttt{TDec} in turn (Lines 11-12). 

Arguably, given $\sum_{i=1}^{n}\frac{\delta_i}{\pi}\omega_i$ ($\pi=\sum_{i=1}^{n}\delta_i$), a participant $i$ that only knows $\omega_i,\delta_i$ fails to learn other participants model parameters. Moreover, SP and CSP jointly calculate the global average model in the ciphertext field. Thus, \texttt{PriFedAvg} does not disclose participants model privacy.

\section{Strategy Against Dropout}
To the best of our knowledge, existing solutions\cite{liu2019boosting,bonawitz2017practical} usually adopt the reconstruction capability of threshold secret sharing to keep robustness against dropped participants. Unfortunately, a solution based on the threshold secret sharing requires more computation and network connections. If a participant may lose a network connection, it may be an unrealistic countermeasure to require the participant to build more network connections to restrain the network connection loss. To relieve this dilemma, we propose discard and retransmission strategies against participant dropouts.

\subsection{Discard Strategy}
The key idea of discard strategy is that even though one participant disconnects, other online participants can still collaboratively produce a global average model.
\begin{definition}[Discard Strategy]
Considering unstable wireless network connections, SP sets a time interval, e.g., $[t_0, t_0+t_\Delta]$. Participants' training models only submitted in the time range $[t_0, t_0+t_\Delta]$ are considered a valid training model. Participants' training models submitted beyond the time range $[t_0, t_0+t_\Delta]$ are discarded.
\end{definition}

Suppose participants' training model collection $\Omega=\{\omega_1,\cdots,\omega_n\}$, where each participant $i$ trains a model $\omega_i$. According to Algorithm \ref{fedavg} (\texttt{PriFedAvg}), $n$ training models can generate a global average model as follow
\begin{equation}
\bar{\omega}^{(n)}=\frac{\sum_{i=1}^{n}\delta_i\cdot\omega_i}{\sum_{i=1}^{n}\delta_i}.
\end{equation}

Without loss of generality, we assume that the participant $n$ drops out during model aggregation. Apparently, $n-1$ participants' training models can still produce a global average model, i.e.,
\begin{equation}
\bar{\omega}^{(n-1)}=\frac{\sum_{i=1}^{n-1}\delta_i\cdot \omega_i}{\sum_{i=1}^{n-1}\delta_i}.
\end{equation}
Let $\Delta=$ Eq. (4)$-$Eq. (3), we have
\begin{equation}
\Delta=\frac{\sum_{i=1}^{n-1}\delta_i\cdot \omega_i}{\sum_{i=1}^{n-1}\delta_i}-\frac{\sum_{i=1}^{n}\delta_i\cdot \omega_i}{\sum_{i=1}^{n}\delta_i}=\frac{\delta_n\cdot(\bar{\omega}^{(n)}-\omega_n)}{\sum_{i=1}^{n-1}\delta_i}.
\end{equation}
As $\bar{\omega}^{(n)}$ is the average model of $n$ participants' training models, $\omega_n$ is close to $\bar{\omega}^{(n)}$. Therefore, when $\frac{\delta_n}{\sum_{i=1}^{n-1}\delta_i}$ is sufficiently large, $\Delta$ will be extremely small. Particularly, if $\omega_n$ is far from $\bar{\omega}^{(n)}$, $n-1$ participants' training models will be closer to $\bar{\omega}^{(n-1)}$ when $\omega_n$ is discarded. Hence, the influence of a discarded training model on the global average model is limited. 

Now, considering $k$ participants may drop out, we can learn that 
\begin{equation}
\bar{\omega}^{(n-k)}=\frac{\sum_{i=1}^{n-k}\delta_i\cdot \omega_i}{\sum_{i=1}^{n-k}\delta_i}. 
\end{equation}
Let $\Delta=$ Eq. (6)$-$Eq. (3), we can derive
\begin{align}
\Delta&=\frac{\sum_{i=1}^{n-k}\delta_i\cdot \omega_i}{\sum_{i=1}^{n-k}\delta_i}-\frac{\sum_{i=1}^{n}\delta_i\cdot \omega_i}{\sum_{i=1}^{n}\delta_i}\nonumber\\
&=\frac{\sum_{i=n-k+1}^{n}\delta_i\cdot (\bar{\omega}^{(n)}-\omega_i)}{\sum_{i=1}^{n-k}\delta_i}.
\end{align}
Since $\bar{\omega}^{(n)}$ is the global average model of $n$ participants' training models, $\omega_i$ ($n-k+1\leq i\leq n$) is close to $\bar{\omega}^{(n)}$. Thus, when $\frac{\sum_{i=n-k+1}^{n}\delta_i}{\sum_{i=1}^{n-k}\delta_i}$ is big enough, $\Delta$ will be very small. Besides, if $k$ dropped participants' training models are far from $\bar{\omega}^{(n)}$, $n-k$ participants' training models are closer to $\bar{\omega}^{(n-k)}$. Thus, discarded training models have limited influences on the global average model.

Take together, it is not difficult for us to observe that even though dropped participants' training models are discarded, online participants can still produce a global average model which is close to that all participants are online to produce.

\subsection{Retransmission Strategy}
The key idea of a retransmission strategy is that a new global average model can be recalculated when a calculated global average model and a retransmitted model are given.

\begin{definition}[Retransmission Strategy]
Considering unstable wireless network connections, SP allows a dropped participant to resubmit a training model $\omega_n$ via connecting again. Specifically, given a global average model $\bar{\omega}^{(n-1)}$ of $n-1$ participants and the reuploaded training model $\omega_n$, SP recalculates a new global average model and sends the new average model to all online participants.
\end{definition}

Without loss of generality, we assume that the participant $n$ drops out but resubmits a training model. According to our proposed \texttt{PriFedAvg}, we can learn that SP has $\llbracket\sum_{i=1}^{n-1}\delta_i\cdot\omega_i\rrbracket$. Given $\langle\llbracket\delta_n\cdot \omega_n\rrbracket, \llbracket\delta_n\rrbracket\rangle$, SP can calculate
\begin{align}
\llbracket\sum_{i=1}^{n}\delta_i\rrbracket&=\llbracket\sum_{i=1}^{n-1}\delta_i\rrbracket\cdot\llbracket\delta_n\rrbracket,\\
\llbracket\sum_{i=1}^{n}\delta_i\cdot \omega_i\rrbracket&=\llbracket\sum_{i=1}^{n-1}\delta_i\cdot \omega_i\rrbracket\cdot\llbracket\delta_n\cdot\omega_n\rrbracket.
\end{align}
After that, SP and CSP can jointly execute $\texttt{SDIV}(\llbracket\sum_{i=1}^{n}\delta_i\cdot \omega_i\rrbracket, \llbracket\sum_{i=1}^{n}\delta_i\rrbracket)$ to obtain $\llbracket\bar{\omega}^{(n)}\rrbracket$. After that, SP calculates $\bar{\omega}^{(n)}_{1}\leftarrow\texttt{PDec}(\llbracket\bar{\omega}^{(n)}\rrbracket, \lambda_{sp})$ and sends $\langle\bar{\omega}^{(n)}, \bar{\omega}_1^{(n)}\rangle$ to all online participants. Once receiving $\langle\bar{\omega}^{(n)}, \bar{\omega}_1^{(n)}\rangle$, participants can calculate $\bar{\omega}_2^{(n)}\leftarrow\texttt{PDec}(\llbracket\bar{\omega}^{(n)}\rrbracket, \lambda_{p})$ and $\bar{\omega}^{(n)}\leftarrow\texttt{TDec}(\bar{\omega}_1^{(n)}, \bar{\omega}_2^{(n)})$. Therefore, participants can obtain a global average model of $n$ training models. Note that if SP has already sent the global average model to participants, the dropped participant's training model is still discarded.

\section{Privacy-preserving Reward Distribution}
Incentive mechanism design is a significant research field in both MCS\cite{karaliopoulos2019optimal,zhao2018federated,qu2020posted}. To achieve the incentive mechanism's truthfulness and fairness, we adopt a “posted pricing” mechanism\cite{zhao2020pace,qu2020posted}, which has been proved that is truthfulness and fairness\cite{singla2013truthful}. The challenge of a “posted pricing” mechanism is to set a reasonable price for each participant.

\textsc{CrowdFL} adopts a federated learning approach to outsources data process to participants, which requires participants to train models multiple times until the global average model converging or reaching the maximum number of iterations $T$. Besides, a participant might drop out in each round of model aggregation, which fails to submit a training model. Thus, in each round model aggregation, our proposed reward distribution mechanism pays rewards for participants. Particularly, if a participant drops out in one round model aggregation, the participant can submit her training model in the next round model aggregation to obtain a reward.

In \textsc{CrowdFL}, a requester takes a final global average model $\bar{\omega}$ as a sensing result. Arguably, if a participant's training model $\omega_i$ is closer to $\bar{\omega}$, the participant's training model $\omega_i$ is more effective. Therefore, the participant should obtain more rewards. In general, the participant collect more sensing data, the participant consumes more resources. Thus, an effective reward distribution should consider the amount of collected sensing data. We propose a reward distribution mechanism as follows
\begin{equation}
\setlength{\nulldelimiterspace}{0pt}
\left\{\begin{IEEEeqnarraybox}[\relax][c]{l's}
	w_i=\frac{\delta_i}{\sum_{i=1}^{n}\delta_i}\cdot\frac{\sum_{i=1}^{n}\mathtt{d}(\omega_i,\bar{\omega})}{\mathtt{d}(\omega_i,\bar{\omega})+\epsilon},\\
	\mu_i=b_t\cdot\frac{w_i}{\sum_{i=1}^{n}w_i},
\end{IEEEeqnarraybox}\right.
\end{equation}
where $w_i$ and $\mu$ are the model weight and the reward of a participant $i$, respectively, and $b_t$ is a budget constraint in one round model aggregation. $\mathtt{d}$ is to quantize the distance between $\omega_i$ and $\bar{\omega}$. Formally, $\mathtt{d}(\omega_i, \bar{\omega})=(\omega_i-\bar{\omega})^2$. $\epsilon$ is a control parameter that ensures $\mathtt{d}(\omega_i,\bar{\omega})+\epsilon\neq 0$ (In \textsc{CrowdFL}, $\epsilon=\frac{1}{10^L}$). From Eq. (10), we can see that a training model is closer to the global average model, the larger the model's weight is and the more the model's reward is. Besides, if $\omega_i=\omega_j$ and $\delta_i>\delta_j$, $\mu_i>\mu_j$.

To prevent privacy leakage during reward distribution, we design a privacy-preserving reward distribution mechanism (\texttt{PriRwd}) based on the proposed \texttt{SMUL}. The goal of \texttt{PriRwd} is to calculate Eq. (10) in the ciphertext field. \texttt{PriRwd} is shown in Algorithm \ref{pprd} whose brief description is given below.

\begin{algorithm}[htb] 
\caption{$\texttt{PriRwd}(\{\llbracket\omega_i\rrbracket,\llbracket\delta_i\rrbracket\}_{i=1}^n, \llbracket\bar{\omega}\rrbracket,\llbracket b_t\rrbracket)\!\rightarrow\!\{\llbracket\mu_i\rrbracket\}_{i=1}^n$} 
\label{pprd} 
\begin{algorithmic}[1]
	\REQUIRE SP has $\{\llbracket\omega_i\rrbracket,\llbracket\delta_i\rrbracket\}_{i=1}^n, \llbracket\bar{\omega}\rrbracket,\llbracket b_t\rrbracket$;
	\FOR{$i=1$ to $n$}
	\STATE SP calculates $\llbracket d_i\rrbracket=\llbracket \omega_i\rrbracket\cdot\llbracket \bar{\omega}\rrbracket^{N-1}$;
	\STATE SP and CSP jointly execute \texttt{SMUL} to obtain $\llbracket d_i^2\rrbracket\leftarrow\texttt{SMUL}(\llbracket d_i\rrbracket, \llbracket d_i\rrbracket)$;
	\STATE SP calculates $\llbracket d_i'^2\rrbracket\leftarrow\llbracket d_i^2\rrbracket\cdot\llbracket 1\rrbracket$, $\llbracket\pi\rrbracket\leftarrow\prod_{i=1}^{n}\llbracket\delta_i\rrbracket$, and $\llbracket\Omega\rrbracket\leftarrow\prod_{i=1}^{n}\llbracket d_i^2\rrbracket$;
	\STATE SP and CSP jointly run \texttt{SMUL} to calculate $\llbracket w_i^\downarrow\rrbracket\leftarrow\texttt{SMUL}(\llbracket\pi\rrbracket, \llbracket d_i'^2\rrbracket)$;
	\ENDFOR 
	\FOR{$i=1$ to $n$}
	\STATE SP and CSP jointly compute $\llbracket w_i^\uparrow\rrbracket\leftarrow\texttt{SMUL}(\llbracket\delta_i\rrbracket,\llbracket\Omega\rrbracket)$, $\llbracket w_i\rrbracket\leftarrow\texttt{SDIV}(\llbracket w_i^\uparrow\rrbracket,\llbracket w_i^\downarrow\rrbracket)$, and $\llbracket b_t\cdot w_i\rrbracket\leftarrow\texttt{SMUL}(\llbracket b_t\rrbracket, \llbracket w_i\rrbracket)$;
	\ENDFOR
	\STATE SP calculates $\llbracket\mathcal{W}\rrbracket\leftarrow(\prod_{i=1}^{n}\llbracket w_i\rrbracket)^{10^L}$;
	\FOR{$i=1$ to $n$}
	\STATE SP and CSP jointly execute \texttt{SDIV} to obtain $\llbracket\mu_i\rrbracket\leftarrow\texttt{SDIV}(\llbracket b_t\cdot w_i\rrbracket, \llbracket\mathcal{W}\rrbracket)$;
	\ENDFOR
\end{algorithmic}
\end{algorithm}
The first \textbf{for} loop (Lines 1-6) is to calculate the distance between $\omega_i$ and $\bar{\omega}$. Specifically, the ciphertexts of $\sum_{i=1}^{n}\mathtt{d}(\omega_i,\bar{\omega})$ and $(\sum_{i=1}^{n}\delta_i)\cdot(\mathtt{d}(\omega_i,\bar{\omega})+\epsilon)$ are calculated through the proposed \texttt{SMUL}, where $\epsilon$ is set as $10^{-L}$. The second \textbf{for} loop (Lines 7-9) is to calculate the ciphertexts of model weights $w_i$ and $b_t\cdot w_i$. Specifically, the algorithm calls the proposed \texttt{SDIV} and \texttt{SMUL} to compute $\llbracket w_i\rrbracket$ and $\llbracket b_t\cdot w_i\rrbracket$, respectively. The last \textbf{for} loop is to calculate $\llbracket\mu_i\rrbracket$ by using the proposed \texttt{SDIV}. 

After obtaining $\llbracket\mu_i\rrbracket$, SP can calculate $\llbracket\mu_{i,1}\rrbracket\leftarrow\texttt{PDec}(\llbracket\mu_i\rrbracket, \lambda_{sp})$ and sends $\langle\llbracket\mu_i\rrbracket, \llbracket\mu_{i,1}\rrbracket\rangle$ to the participant $i$. Once receiving $\langle\llbracket\mu_i\rrbracket, \llbracket\mu_{i,1}\rrbracket\rangle$, the participant $i$ can compute $\llbracket\mu_{i,2}\rrbracket\leftarrow\texttt{PDec}(\llbracket\mu_i\rrbracket, \lambda_p)$. Finally, the participant can obtain her reward $\mu_i$ by using the \texttt{TDec} to decrypt $\langle\llbracket\mu_{i,1}\rrbracket,\llbracket\mu_{i,2}\rrbracket\rangle$. From Algorithm \ref{pprd}, we can see that only the participant can obtain her reward $\mu_i$. Particularly, to prevent the requester from refusing to pay rewards for participants, the \texttt{PriRwd} requires the requester to pay rewards to SP in advance.

\section{Privacy Analysis}
In this section, we firstly demonstrate that our proposed \texttt{SDIV} and \texttt{SMUL} do not leak input data privacy. Then, we prove that the proposed \texttt{PriFedAvg} and \texttt{PriRwd} can protect model privacy and reward privacy.

\begin{theorem}
In \texttt{SDIV}, given $\llbracket x\rrbracket$ and $\llbracket y\rrbracket$, SP and CSP without collusion cannot learn $x$, $y$, and $x\div y$.
\end{theorem}
\begin{proof}
{\rm As the Paillier cryptosystem is semantic-security, given $\llbracket x\rrbracket$, $\llbracket y\rrbracket$, and $\llbracket x\div y\rrbracket$, SP without a private key $\lambda$ fails to obtain $x$, $y$, and $x\div y$.
	
	CSP can obtain the intermediate results of computations, i.e., $rx+(r\alpha+e)\cdot y$, $ry$, and $x\div y+\alpha+e\div y$. However, as there is no a known efficient, non-quantum integer factorization algorithm to factorize a sufficiently larger composite number, CSP cannot factorize $ry$ into a product of smaller integers. Formally, $\Pr[ry=q_1\cdot q_2\cdots q_m|ry]<\varepsilon$, where $\varepsilon$ is a negligible probability. Even though CSP can learn a factor $q_i$, he cannot determine whether $q_i$ is a factor of $r$ or that of $y$. In other words, even if CSP can learn $q_1, q_2,\cdots, q_m$, he fails to obtain the factorization of $y$, i.e., $y=q_1'\cdot q_2'\cdots q_n'$ ($q_i'\in\{q_1,q_2,\cdots,q_m\}$). Thus, given $ry$, CSP has no knowledge of $y$. Similarly, CSP cannot learn $r$. 
	
	Given $rx+(r\alpha+e)\cdot y$ and $x\div y+\alpha+e\div r$, as $(r\alpha+e)\cdot y$ and $\alpha+e\div r$ are random, and the binary length of $(r\alpha+e)\cdot y$ is larger than $rx$ as well as that of $\alpha+e\div r$ is lager than $x\div y$, $rx+(r\alpha+e)\cdot y$ and $x\div y+\alpha+e\div r$ are indistinguishable from random
	values. As a matter of fact, when $(r\alpha+e)\cdot y$ and $\alpha+e\div r$ are sufficiently larger, in the view of CSP, $rx+(r\alpha+e)\cdot y$ and $x\div y+\alpha+e\div r$ are encrypted by one-time pad. Hence, CSP has no knowledge of $rx$ and $x\div y$.
	
	Taken together, SP and CSP without collusion cannot learn $x$, $y$, and $\frac{x}{y}$.
}
\end{proof}

\begin{theorem}
In \texttt{SMUL}, given $\llbracket x\rrbracket$ and $\llbracket y\rrbracket$, SP and CSP without collusion cannot learn $x$, $y$, and $xy$.
\end{theorem}
\begin{proof}
{\rm As the Paillier cryptosystem is semantic-security, given $\llbracket x\rrbracket$, $\llbracket y\rrbracket$, and $\llbracket xy\rrbracket$, SP without private key $\lambda$ fails to learn $x$, $y$, and $xy$.
	
	CSP can obtain the intermediate computation results, i.e., $x+r_1$ and $y+r_2$. However, as long as $\sigma>\ell$, in the view of CSP, $x+r_1$ and $y+r_2$ are indistinguishable from random
	values. In essence, $x+r_1$ and $y+r_2$ are encrypted by one-time pad. One-time pad has been proved that is perfectly secret\cite{katz2020introduction}. Therefore, even though CSP can obtain $x+r_1$, $y+r_2$, and $(x+r_1)\cdot(y+r_2)$, he has no knowledge of $x$, $y$, and $xy$.
	
	Taken together, SP and CSP without collusion cannot learn $x$, $y$, and $xy$.
}
\end{proof}

\begin{theorem}
The proposed \texttt{PriFedAvg} does not leak a participant's training model to other entities.
\end{theorem}
\begin{proof}
{\rm As the proposed \texttt{SDIV} can prevent SP and CSP from learning input data and the calculation result, SP and CSP fail to obtain a participant's training model.
	
	Each participant can obtain $\bar{\omega}=\frac{\sum_{i=1}^{n}\delta_i\cdot \omega_i}{\sum_{i=1}^{n}\delta_i}$, however, as the participant fails to learn $\sum_{i=1}^{n}\delta_i\cdot \omega_i$ and $\sum_{i=1}^{n}\delta_i$, the participant cannot learn other participants' training models. Similarly, although a requester can get $\bar{\omega}$, the requester cannot learn participants' training models.
	
	Taken together, no entity except for the participant herself can learn the participant's training model.
}
\end{proof}

\begin{theorem}
The proposed \texttt{PriRwd} does not leak a participant reward.
\end{theorem}
\begin{proof}
{\rm According to Theorem 1 and Theorem 2, the \texttt{SDIV} and the \texttt{SMUL} do not leak input data and calculation results. Thus, SP and CSP fail to obtain a participant reward.
	
	Besides, the participant is only given $\langle\llbracket\mu_i\rrbracket, \llbracket\mu_{i,1}\rrbracket\rangle$, hence, the participant only obtains her reward.
	
	Taken together, only the participant can learn her reward.
}
\end{proof}

\section{Performance Evaluation}
In this section, we evaluate the effectiveness and efficiency of the proposed \textsc{CrowdFL} from the views of theoretical analysis and experimental testing.

\textbf{Dataset.} Considering a practical MCS application that recognizes human activity through MCS\cite{lyu2017privacy}, we utilize Actitracker dataset\cite{lockhart2012applications} released by WISDM lab. The sensing data is collected by 36 participants using smartphones in their pockets. Moreover, participants are responsible for processing collected data in end-devices. In our experiments, we set 36 participants according to the Actitracker dataset, and each participant owns 13230 data records (each data record represents one acceleration in three spatial coordinates, and each activity consists of 270 data records). To evaluate the effectiveness, participants' sensing data is set as training data, which accounts for 70\% of the total experimental data. The remaining 30\% of data is used as test data. The distribution of six activities is shown in Table 2, and each activity depends on the acceleration of three spatial coordinates. As illustrated in Table 2, the training data is unbalanced.
\begin{table}
\label{data}
\centering
\resizebox{0.49\textwidth}{!}{
	\begin{threeparttable}
		\caption{The Distribution of Training Data}
		\begin{tabular}{cccccc}
			\toprule
			Jogging & Walking & Upstairs & Downstairs & Sitting & Standing \\ \midrule
			518   &   693   &   200    &    176     &   92    &    85    \\ \bottomrule
		\end{tabular}
	\end{threeparttable}
}
\end{table}

\textbf{Configuration.}  We implement \textsc{CrowdFL} in Java and adopt Deeplearning4j\footnote{https://deeplearning4j.org/} to train a three-layer neural network. We utilize a Windows 10 desktop, with an Intel(R) Core(TM) i7-8700 CPU @3.20GHz and 16.0GB RAM, to serve as SP and CSP. Moreover, we deploy an Android 10 smartphone, with HUAWEI Kirin 980 and 8.0GB RAM, to simulate 36 participants. We set the security parameter of the PCTD as $\zeta=1024$ and use a Shamir's Secret Sharing algorithm over GF(256) from Maven Repository to split training models. We set up parameters in \textsc{CrowdFL} as follows. The learning rate is $\eta=0.001$, the batch size is $B=50$, epochs are $E=30$, the number of input nodes, output nodes, and hidden nodes is $270$, $6$, and $1000$, respectively.
\subsection{Effectiveness Analysis}
In this section, we demonstrate the proposed \textsc{CrowdFL} is effective. The \textsc{CrowdFL} consists of the \texttt{PriFedAvg}, strategies against participant dropouts, and the \texttt{PriRwd}. Arguably, if all of the components of the \textsc{CrowdFL} are effective, the \textsc{CrowdFL} can be considered as effective.

To demonstrate the effectiveness of the \texttt{PriFedAvg}, we firstly compare the \texttt{PriFedAvg}, non-federated neural network (NN), and \texttt{FedAvg} from accuracy, precision, recall, and F1 Score. The comparison results are shown in Table 3. Table 3 shows that the difference between \texttt{FedAvg} and NN reduces through multiple interactions. When the number of iterations reaches a certain number (e.g., 5), the accuracy difference is lower than 0.01, and the difference in the recall is lower than 0.03. The difference between the MAD (Mean Absolute Deviation) of \texttt{FedAvg}\_5 and that of NN is less than 0.03. Moreover, from Table 3, it can be observed that our proposed \texttt{PriFedAvg} has the same results in accuracy, precision, recall, and F1 Score with the \texttt{FedAvg} when the \texttt{PriFedAvg} is iterated five times. Thus, we can conclude that the \texttt{PriFedAvg} is as effective as the \texttt{FedAvg}.
\begin{table}
\centering
\resizebox{0.49\textwidth}{!}{
	\begin{threeparttable}
		\caption{Comparison between NN and CrowdFL}
		\begin{tabular}{lccccc}
			\toprule
			Schemes            &    Accuracy     &    Precision    &     Recall      &    F1 Score     & MAD    \\ \midrule
			NN                 &     0.7778      &     0.6891      &     0.6961      &     0.6825      & --     \\
			\texttt{FedAvg}\_1 &     0.6481      &     0.6358      &      0.574      &     0.6144      & 0.0933 \\
			\texttt{FedAvg}\_2 &     0.7553      &     0.6497      &      0.642      &     0.5733      & 0.0563 \\
			\texttt{FedAvg}\_3 &     0.7619      &     0.6311      &     0.6605      &     0.6077      & 0.0461 \\
			\texttt{FedAvg}\_4 &     0.7725      &     0.6335      &     0.6653      &     0.6186      & 0.0389 \\
			\texttt{FedAvg}\_5 & \textbf{0.7751} & \textbf{0.6446} & \textbf{0.6757} & \textbf{0.6390} & \textbf{0.0278} \\
			\texttt{PriFedAvg}   & \textbf{0.7751} & \textbf{0.6446} & \textbf{0.6757} & \textbf{0.6390} & \textbf{0.0278} \\ \bottomrule
		\end{tabular}
		\begin{tablenotes}
			\small
			\item \textbf{Note.} MAD: Mean Absolute Deviation. The subscript number of \texttt{FedAvg} is the number of iterations. \texttt{PriFedAvg} iterates five times.
		\end{tablenotes}
\end{threeparttable}}
\end{table}

Figs. 3(a) and 3(b) show the results of comparison of accuracy and F1 score among NN, \textsc{CrowdFL-D} (\textsc{CrowdFL} with a discard strategy), and \textsc{CrowdFL-R} (\textsc{CrowdFL} with a retransmission strategy) in different dropout rates. Particularly, we assume that training models of 50\% dropped participants are retransmitted in \textsc{CrowdFL-R}. As depicted in Fig. 3, we can learn that \textsc{CrowdFL} has better performance in accuracy and F1 score than non-federated NN. Moreover, as dropped participants' training models in \textsc{CrowdFL-R} are retransmitted, \textsc{CrowdFL-R} is generally more accurate than \textsc{CrowdFL-D}, and has a larger F1 score than \textsc{CrowdFL-D}. According to the results of Fig. 3, it is not difficult to conclude that the proposed strategies against participant dropouts in \textsc{CrowdFL} are feasible. Furthermore, the \textsc{CrowdFL-R} might be more robust against dropped participants than the \textsc{CrowdFL-D}.
\begin{figure}
\centering
\subfigure[Comparison of accuracy]{
	\includegraphics[width=0.445\linewidth]{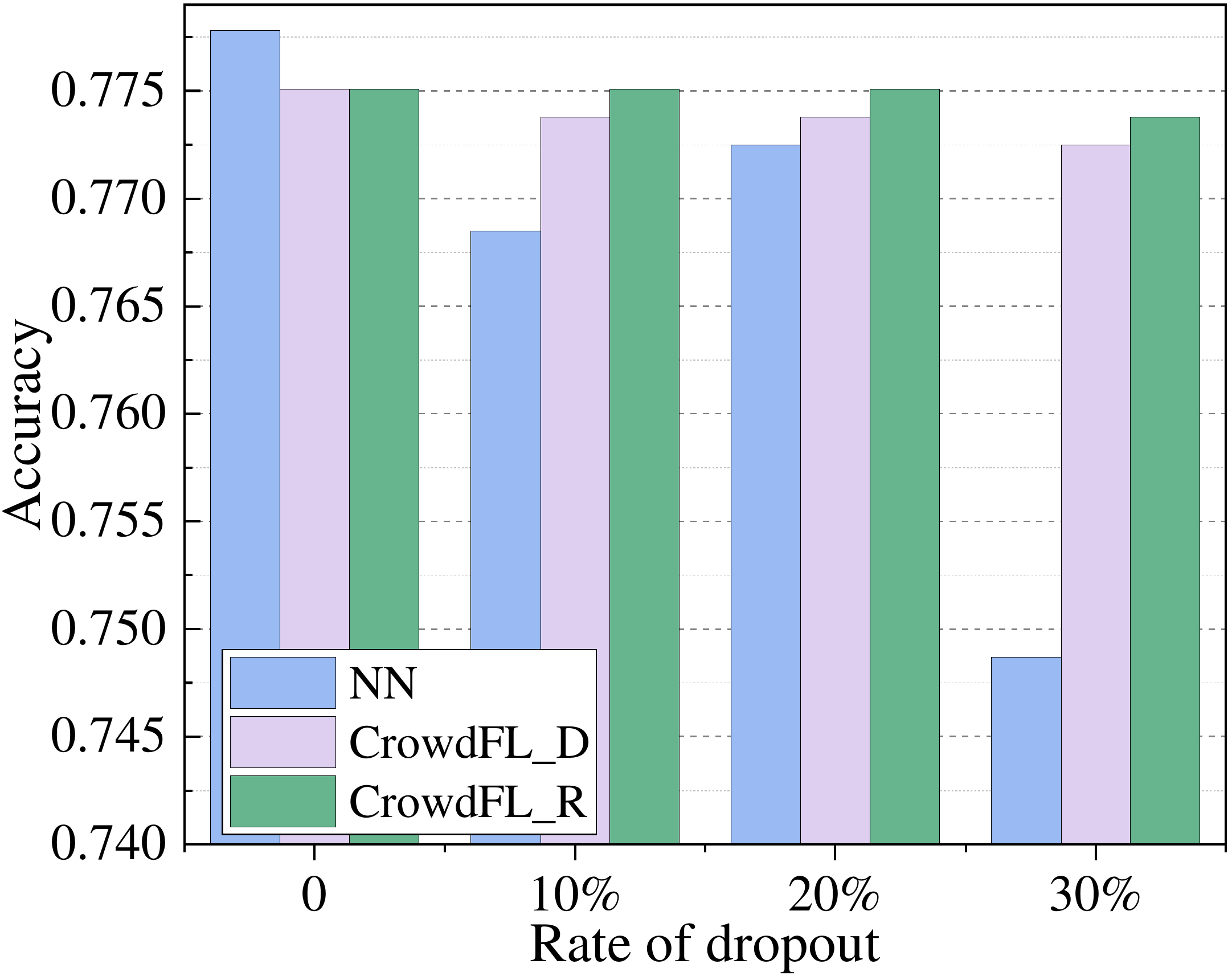}}
\hspace{0.1in}
\subfigure[Comparison of F1 score]{
	\includegraphics[width=0.44\linewidth]{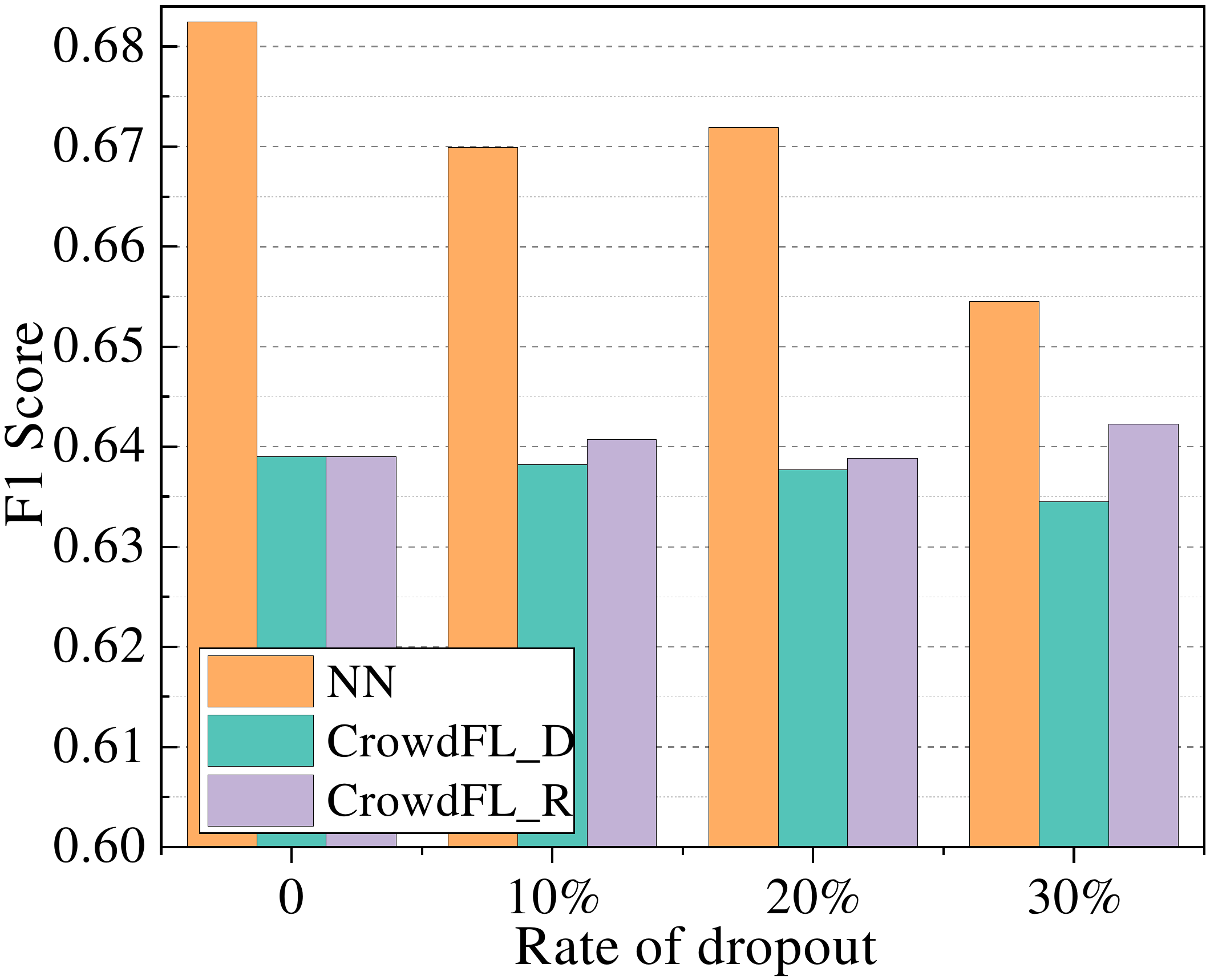}}
\caption{Comparison of accuracy and F1 score in different dropout rates.}
\end{figure}

If an average model is considered as an effective model, a participant's training model is closer to the average model and is more efficient. In our proposed reward distribution mechanism, if a participant's training model is closer to the average model (i.e., $\mathtt{d}(\omega_i,\bar{\omega})$ is less), the participant's model weight (i.e., $w_i$) is larger and receives more rewards (i.e., $\mu_i$). Fig. 4 shows experimental results of reward distributions. Particularly, we randomly choose the second iteration and the fourth iteration to illustrate our proposed reward distribution mechanism's feasibility, and the incentive budget of each iteration is set as $b_2=b_4=\$36$. According to Figs. 4(a) and 4(b), we can see that the less $\mathtt{d}(\omega_i,\bar{\omega})$ is, the larger $w_i$ is. Besides, multiple iterations can reduce $\mathtt{d}(\omega_i,\bar{\omega})$ meaning that a training model is closer to an average model, the training model is more effective. From Figs. 4(b) and 4(c), we can learn that $\mu_i$ is positively correlated with $w_i$. In other words, the larger $w_i$ is, the more $\mu_i$ is. Taken together, it can be concluded that the proposed reward distribution mechanism can motivate participants to train a more effective model.
\begin{figure}
\centering
\subfigure[$w_i$ vs $\mathtt{d}(\omega_i,\bar{\omega})$ in the second iteration]{
	\includegraphics[width=0.45\linewidth]{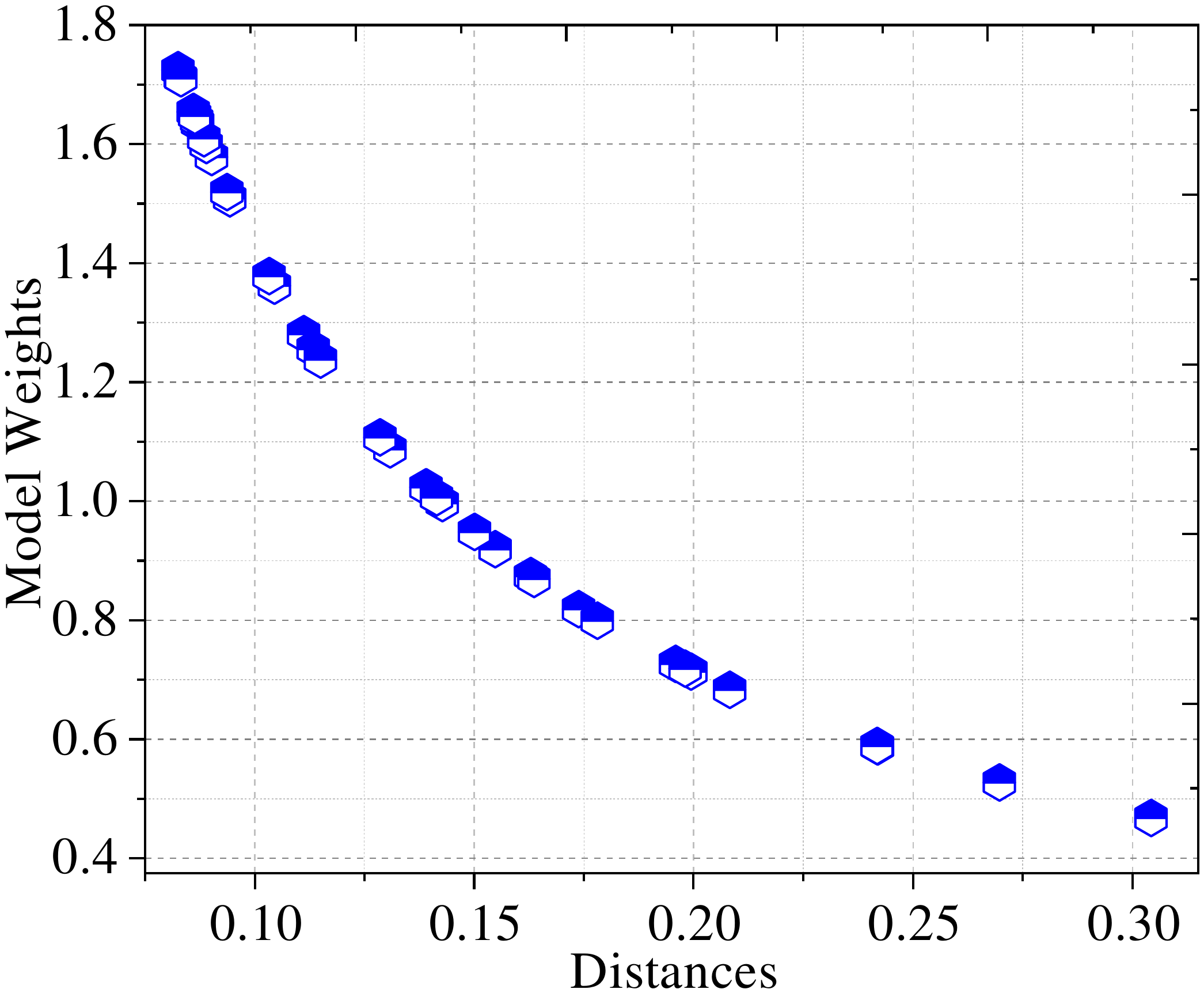}}
\hspace{0.1in}
\subfigure[$\mu_i$ vs $w_i$ in the second iteration]{
	\includegraphics[width=0.45\linewidth]{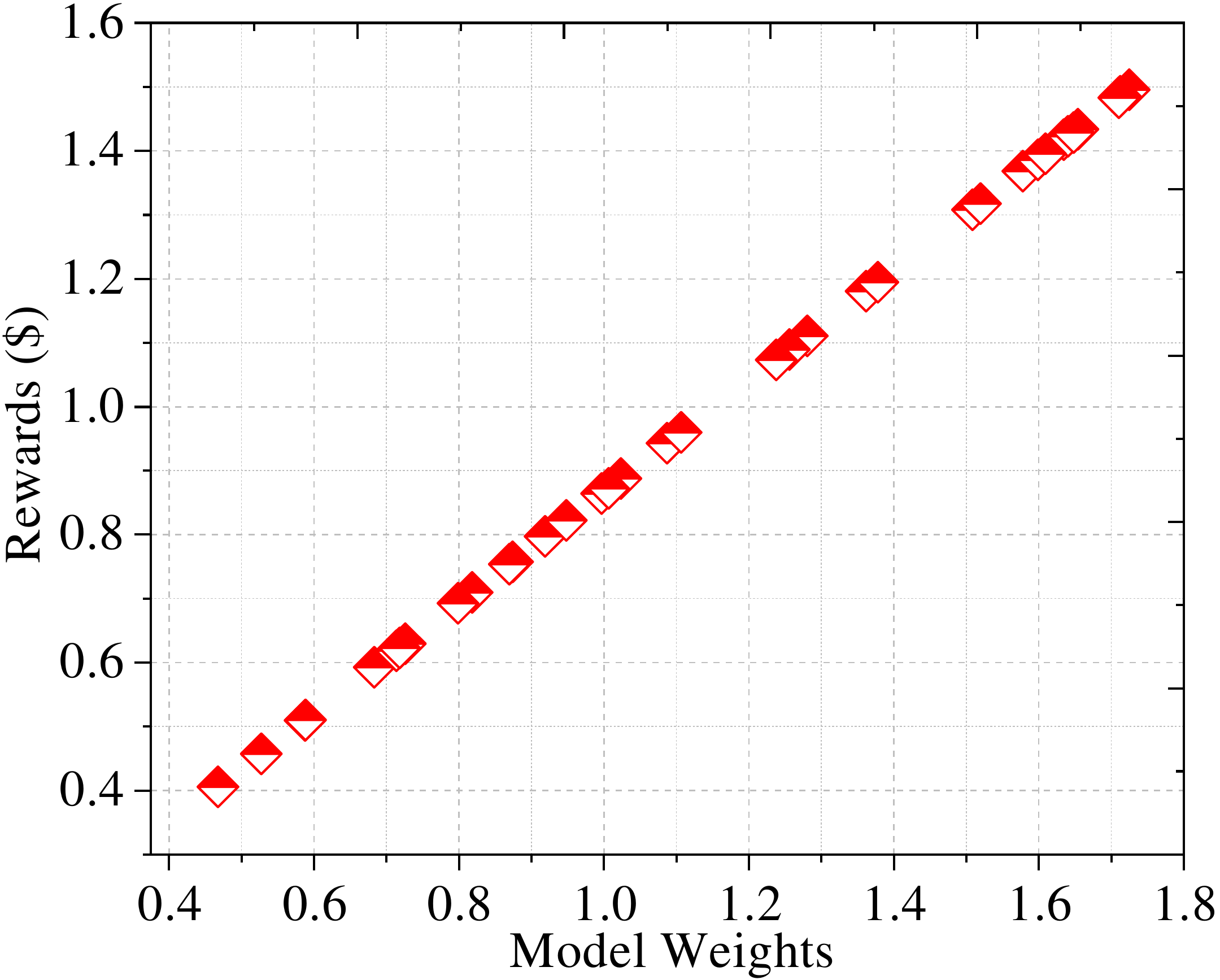}}
\hspace{0.125in}
\subfigure[$w_i$ vs $\mathtt{d}(\omega_i,\bar{\omega})$ in the fourth iteration]{
	\includegraphics[width=0.45\linewidth]{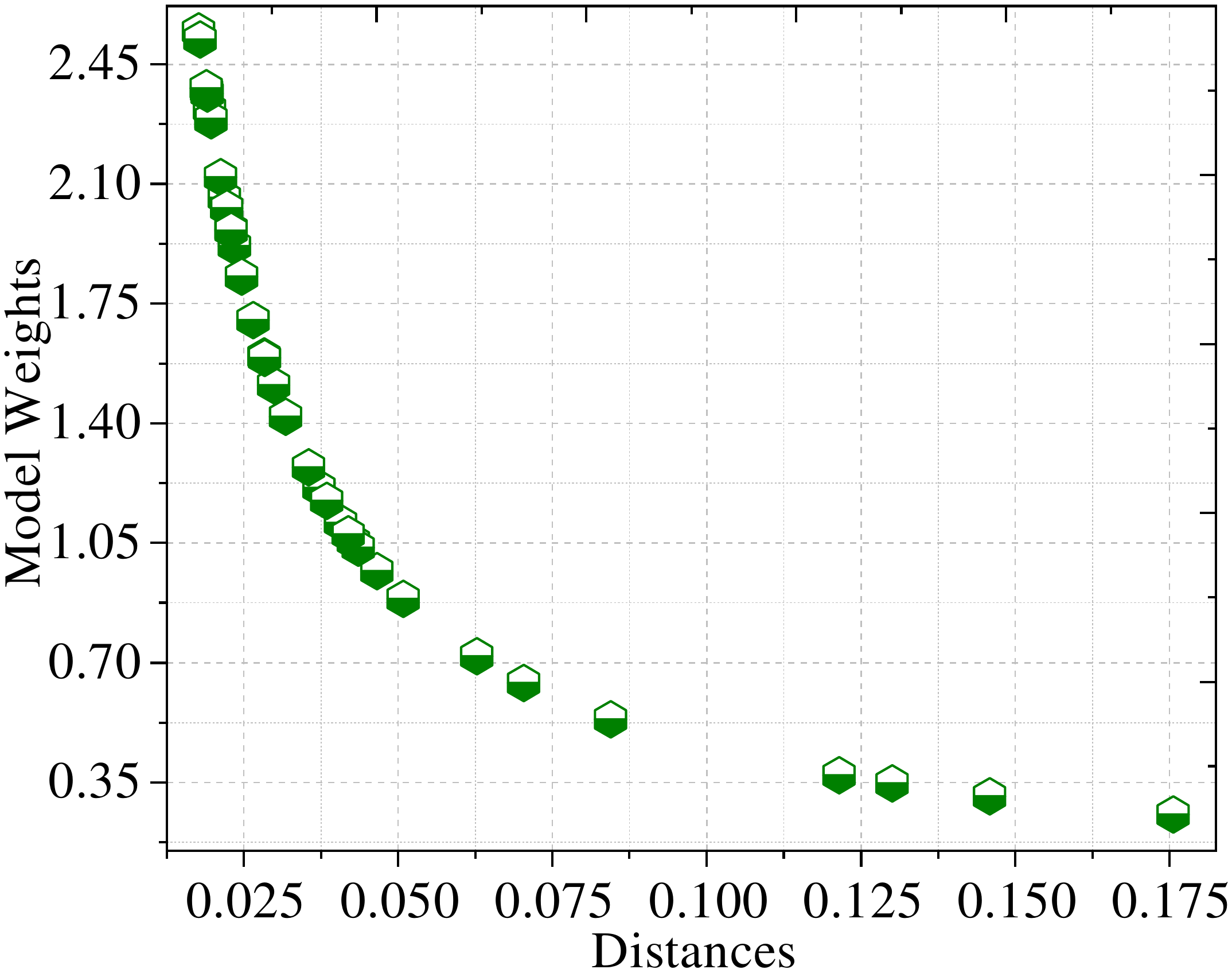}}
\hspace{0.125in}
\subfigure[$\mu_i$ vs $w_i$ in the fourth iteration]{
	\includegraphics[width=0.45\linewidth]{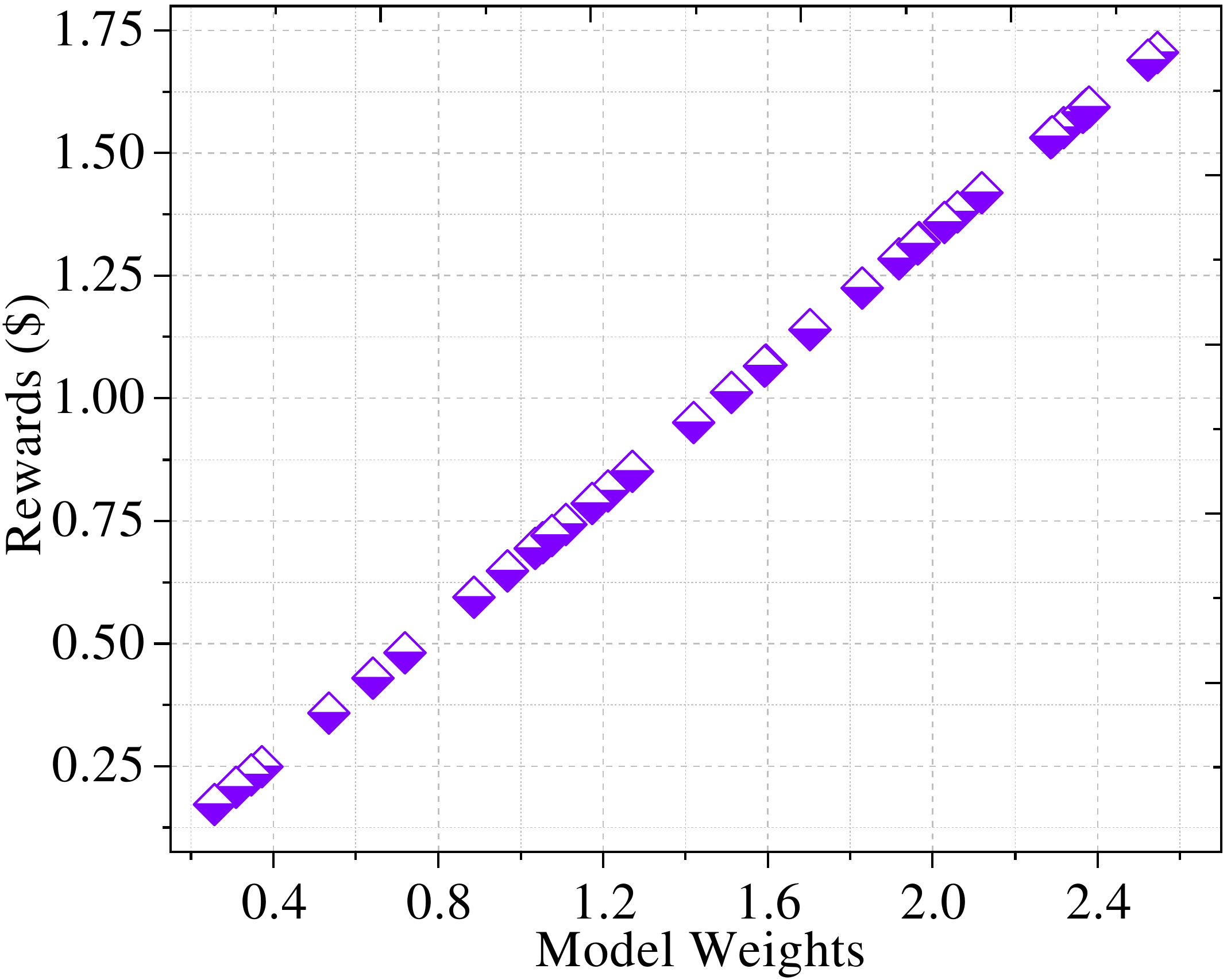}}
\caption{Incentive mechanism.}
\end{figure}

\subsection{Efficiency Analysis}
\textit{1) Theoretically Analysis:} Table 4 shows the comparison between \textsc{CrowdFL} and PPFDL\cite{xu2020privacy}. Compared to PPFDL, our proposed \textsc{CrowdFL} provides privacy-preserving federate average (\texttt{PriFedAvg}) and privacy-preserving reward distribution (\texttt{PriRwd}). Similar to PPFDL, the \textsc{CrowdFL} adopts two servers to support privacy-preserving training models aggregation, but the difference between PPFDL and \textsc{CrowdFL} is that any server in \textsc{CrowdFL} does not has a private key of a requester. From Algorithm \ref{fedavg}, we can see that the proposed \texttt{PriFedAvg} requires SP to execute $2n-2$ times \texttt{PAdd} to obtain $\llbracket\sum_{i=1}^{n}\delta_i\cdot \omega_i\rrbracket$ and $\llbracket\sum_{i=1}^{n}\delta_i\rrbracket$. After that, SP and CSP jointly perform the proposed \texttt{SDIV} to obtain $\llbracket\frac{\sum_{i=1}^{n}\delta_i\cdot \omega_i}{\sum_{i=1}^{n}\delta_i}\rrbracket$. Although the computation cost of a participant in \texttt{PriFedAvg} is lightly larger than that of a participant in PPFDL, the \texttt{PriFedAvg} avoids executing secure multiplication protocols between SP and CSP, which reduces computation costs of SP and CSP. Besides, PPFDL needs perform the \texttt{SecDiv} for multiple times, while the \texttt{PriFedAvg} only requires to execute the \texttt{SDIV} for once. The \texttt{SecDiv} requires not only needs two servers, but also two servers to build a garbled circuit. Contrasted to the \texttt{SecDiv}, the \texttt{SDIV} only requires two servers to execute secure computations. Thus, it can be said that our proposed \texttt{SDIV} is more efficient. From Table 4, we can observe that the computation cost of \texttt{PriFedAvg} and \texttt{PriRwd} is even lower that of PPFDL. One possible explanation is that PPFDL requires to multiple iterations to generate an average model.
\begin{table*}
\newcommand{\tabincell}[2]{\begin{tabular}{@{}#1@{}}#2\end{tabular}}
\centering
\resizebox{\textwidth}{!}{%
	\begin{threeparttable}
		\caption{Comparison of Computation Overhead between \textsc{CrowdFL} and PPFDL}
		\begin{tabular}{@{}cccc@{}}
			\toprule
			\multirow{2}{*}{Scheme} & \multicolumn{3}{c}{Computation cost}  \\ \cmidrule(l){2-4} 
			& SP & CSP & Participants      \\ \midrule
			PPFDL\cite{xu2020privacy}                   & \tabincell{c}{$\mathcal{K}((4n-3)\cdot\texttt{PAdd}+2n\texttt{SMUL}+(n+1)\cdot\texttt{SecDiv}$\\$+n\texttt{PMul})+(n-1)\texttt{PAdd}+\texttt{PMul}$}   &  $\mathcal{K}(2n\texttt{SMUL}+(n+1)\cdot\texttt{SecDiv})$   &        $\texttt{Enc}+\texttt{Dec}$           \\ \midrule
			\texttt{PriFedAvg}               & $(2n-2)\texttt{PAdd}+\texttt{SDIV}$  & \texttt{SDIV}   & $2\texttt{Enc}+\texttt{PMul}+\texttt{PDec}+\texttt{TDec}$                \\ \midrule
			\texttt{PriRwd}                 &  $(4n-2)\cdot\texttt{PAdd}+n\texttt{PMul}+4n\texttt{SMUL}+2n\texttt{SDIV}$  &   $4n\texttt{SMUL}+2n\texttt{SDIV}$  &        \texttt{Enc}           \\ \bottomrule
		\end{tabular}%
		\begin{tablenotes}
			\small
			\item \textbf{Note.} $n$ is the number of participants. \texttt{PAdd} and \texttt{PMul} denotes one Paillier homomorphic addition (i.e., $\llbracket x+y\rrbracket\leftarrow\llbracket x\rrbracket\cdot\llbracket y\rrbracket$) and one Paillier homomorphic multiplication (i.e., $\llbracket xy\rrbracket\leftarrow\llbracket y\rrbracket^x$), respectively. \texttt{SecDiv} is a secure multiplication protocol proposed by PPFDL\cite{xu2020privacy}. $\mathcal{K}$ denotes the iteration times of PPFDL.
		\end{tablenotes}
	\end{threeparttable}
}
\end{table*}

Suppose the binary length of a ciphertext encrypted by the Paillier cryptosystem be $\texttt{L}$, Table 5 compares communication overhead and communication round between \textsc{CrowdFL} and PPFDL. Comparing PPFDL with \texttt{PriFedAvg}, it is easy to see that the communication cost and communication round of two servers in \textsc{CrowdFL} are less than that of two servers in PPFDL. Note that the communication cost of a participant in \textsc{CrowdFL} is slightly higher than that of a participant in PPFDL. One possible explanation is that the participant in \textsc{CrowdFL} is responsible for computing $\llbracket\delta_i\cdot\omega_i\rrbracket$ instead of SP and CSP. The participant who knows $\delta_i$ and $\omega_i$ has less computation cost to obtain $\llbracket\delta_i\cdot\omega_i\rrbracket$ than SP only knowing $\llbracket\delta_i\rrbracket$ and $\llbracket\omega_i\rrbracket$. During reward distribution (see Algorithm \ref{pprd}), as \texttt{PriRwd} needs to execute secure computation protocols multiple times, communication cost and communication round of SP and CSP increase. Even so, when $\mathcal{K}>3$, the communication cost and communication round of PPFDL are larger than that of \textsc{CrowdFL}.
\begin{table*}[]
\newcommand{\tabincell}[2]{\begin{tabular}{@{}#1@{}}#2\end{tabular}}
\centering
\resizebox{\textwidth}{!}{%
	\begin{threeparttable}
		\caption{Comparison of Communication Overhead between \textsc{CrowdFL} and PPFDL}
		\begin{tabular}{@{}ccccccc@{}}
			\toprule
			\multirow{2}{*}{Scheme} &                 \multicolumn{3}{c}{Communication cost}                  & \multicolumn{3}{c}{Communication round} \\
			\cmidrule(l){2-7}    &         SP         &      CSP      &  \multicolumn{1}{c|}{Participant}  & SP & CSP &         Participants         \\ \midrule
			PPFDL\cite{xu2020privacy}          &    $\mathcal{K}((11n+3)\texttt{L}+(n+1)(\|\mathtt{GC}\|+5\texttt{h}))$                &       $\mathcal{K}((9n+3)\texttt{L}+(n+1)(\|\mathtt{GC}\|+5\texttt{h}))$        &       \multicolumn{1}{c|}{$2\texttt{L}$}       &  $\mathcal{K}(10n+4)$  &  $\mathcal{K}(8n+4)$   &       2                       \\ \midrule
			\texttt{PriFedAvg}    & $(4n+5)\texttt{L}$ & $5\texttt{L}$ & \multicolumn{1}{c|}{$4\texttt{L}$} & $2n+2$  &  2  &              2               \\ \midrule
			\texttt{PriRwd}     &      $31n\texttt{L}$              &    $30n\texttt{L}$           &       \multicolumn{1}{c|}{\texttt{L}}        & $13n$   &  $12n$   & $1$                             \\ \bottomrule
		\end{tabular}%
		\begin{tablenotes}
			\small
			\item \textbf{Note.} $n$ is the number of participants. \texttt{L} is the binary length of a ciphertext encrypted by the Paillier cryptosystem. \texttt{h} is the binary length of garbled values in the Garbled Circuits. $\|\mathtt{GC}\|$ indicates the binary length of the Garbled Circuits that takes $x+h_x$, $y+h_y$, $r$, $h_x$, and $h_y$ as inputs, and outputs $x/y-r$.
		\end{tablenotes}
	\end{threeparttable}
}
\end{table*}

To keep robustness against participant dropouts, traditional solutions usually employ secret sharing, such as Shamir secret sharing, which is able to recover a secret through multiple shares. In essence, it is to recover the secret through redundant information. Formally speaking, suppose a secret $s$, to recover $s$, $s$ is firstly split multiple shares (i.e., $n$ shares). Then, $s$ can be recovered by $t$ of $n$ shares (in this paper, we set $(t, n)=(3,5)$). Thus, the computation overhead of PPFL\cite{bonawitz2017practical} that adopts secret sharing to restrain dropped participants, is to execute one secret sharing. Besides, PPFL\cite{bonawitz2017practical} requires a participant to communicate with $n$ other participants and share shares with the server to recover a model. In contrast to PPFL, a dropped participant's training model in the proposed \textsc{CrowdFL} is either discarded or retransmitted. The participant does not need to execute extra computation. The sensing platform requires to perform twice homomorphic addition and one \texttt{SDIV}. When \textsc{CrowdFL} utilizes the discard strategy, there is no additional communication between a participant and the requester. When the retransmission strategy is adopted, a participant needs to resubmit her training model to SP. Meanwhile, if SP receives the dropped participant's training model before sending the encrypted global average model to participants, there is still no additional communication for the SP. The above analyses are listed in Table 6.
\begin{table}
\centering
\renewcommand\arraystretch{1.15}
\newcommand{\tabincell}[2]{\begin{tabular}{@{}#1@{}}
		#2
\end{tabular}}
\begin{threeparttable}
	\caption{Comparison of Computation and Communication Overhead against Participants Dropout}
	\begin{tabular}{lc|ccc}
		\toprule
		\multirow{2}{*}{Schemes}           &                 \multirow{2}{*}{Computation overhead}                            & \multicolumn{3}{c}{Communications round} \\ \cline{3-5}
		& & Participant & SP        & CSP        \\ \midrule
		\textsc{CrowdFL-D} & -- & 0 & 0 & 0\\
		\textsc{CrowdFL-R} & 12.3 \texttt{S} & 1 & 2 & 2\\
		PPFL\cite{bonawitz2017practical} & 3.4 \texttt{S} & $m+1$ & $t$ & --\\\bottomrule
	\end{tabular}
\end{threeparttable}
\end{table}

\textbf{Discussion.} To restrain a participant dropout, it is unadvisable to require the participant to build more connections are unadvisable. Firstly, if a participant may drop out, the participant is likely to fail to establish multiple connections with other participants due to an unstable connection. Besides, if a participant can maintain multiple stable connections with other participants, it is more feasible for the participant to establish one stable connection with SP. Furthermore, the latter requires less computation and communication overhead. 

\textit{2) Experimental Results:} Next, we give experimental results to illustrate the efficiency of \textsc{CrowdFL}. In the above experiment setting, time comparisons of the training model and encrypting model between \textsc{CrowdFL} and PPFDL are shown in Fig. 5, where the suffix ``\_I" means the initial training, and the suffix ``\_A" indicates the training again given an average model. Fig. 5(a) shows that the time of initial training is less than that of training again. One possible explanation is that the latter requires to load the average model. Besides, for the same model parameters and training machine, \textsc{CrowdFL} and PPFDL have the same training time. As depicted in Fig. 5(b), we can observe that encrypting model time of \textsc{CrowdFL} is almost the same as that of PPFDL. Compared to PPFDL, \textsc{CrowdFL} only requires to extra encrypt $\delta_i$, i.e., $\llbracket\delta_i\rrbracket$. The amount of extra encryption is much less the number of training model parameters to be encrypted. Thus, as the Paillier cryptosystem is efficient, the increased encryption time of \textsc{CrowdFL} is almost negligible.
\begin{figure}
\centering
\subfigure[Training model time]{
	\includegraphics[width=0.45\linewidth]{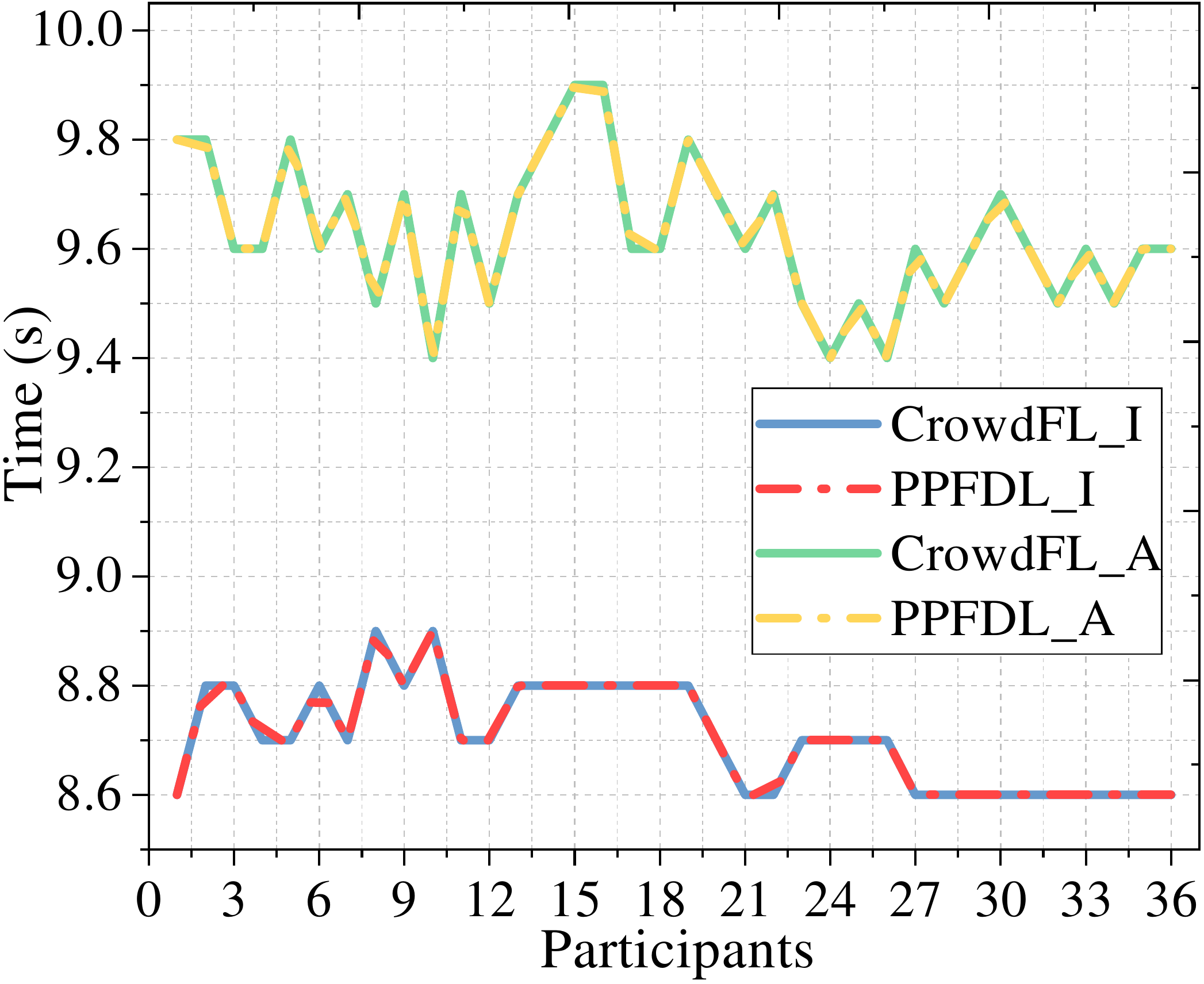}}
\hspace{0.05in}
\subfigure[Encrypting model time]{
	\includegraphics[width=0.45\linewidth]{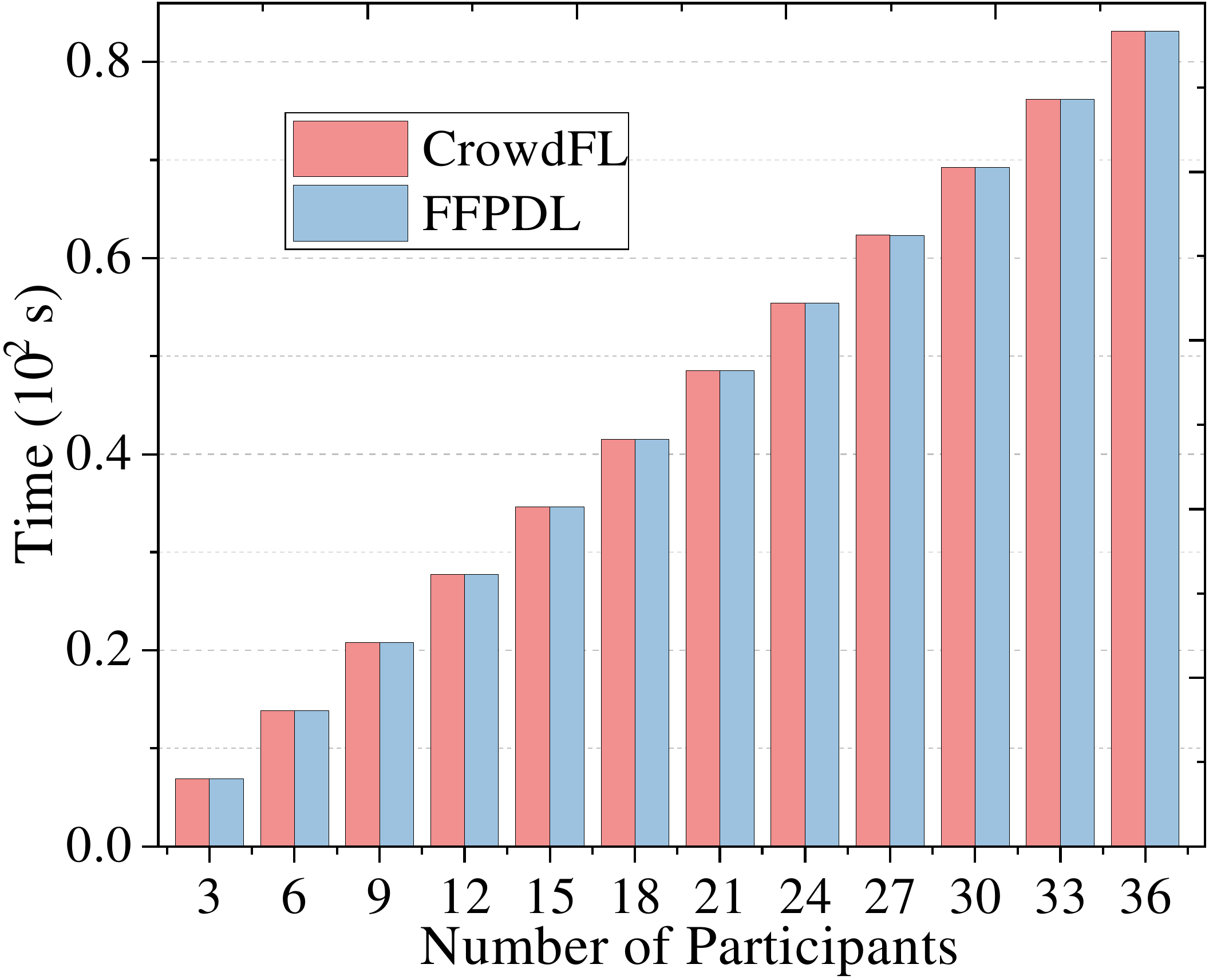}}
\caption{Comparisons of training model time and encrypting model time.}
\end{figure}

As depicted in Fig. 6, we can see comparison results of averaging model time and decrypting model time between \textsc{CrowdFL} and PPFDL\cite{xu2020privacy}. To be fair in comparison, we set $\mathcal{K}=1$ and use \texttt{SDIV} to instead of \texttt{SecDiv}. From Fig. 6(a), it is easy to observe that the averaging model time of \textsc{CrowdFL} is far less than that of PPFDL, which is consistent with the results of theoretical analysis in Table 4. Unlike the \texttt{FedAvg}, PPFDL requires to calculate participants' weights before averaging training models, which increases the computational cost of averaging models. Thus, it can be said that \textsc{CrowdFL} is more efficient than PPFDL to average training models. Fig. 6(b) shows decrypting average model time, where ``\textsc{CrowdFL}\_small" denotes that the participant takes $\lambda_{p}=2$ as the decryption parameter to decrypt the average model. From Fig. 6(b), we can see that if $\lambda_{sp}\in[1, \lambda u]$ and $\lambda_{p}=\varepsilon-\lambda_{sp}$, the decrypting model time of \textsc{CrowdFL} is about double that of PPFDL. One possible explanation is that the binary length of $\lambda_{p}$ is double that of $\lambda$. Although \textsc{CrowdFL} requires the participant to decrypt the average model through using \texttt{PDec} and \texttt{TDec}, the computation cost of \texttt{PDec} and \texttt{TDec} is same with \texttt{Dec}, i.e., executing one \texttt{PMul}, one multiplication, one subtraction, and one division. In \textsc{CrowdFL}, to protect privacy, the participant does not own the private key, which prevents one participant from filching other participants' encrypted models. According to PCTD, we can learn that as long as $\varepsilon=\lambda_{p}+\lambda_{sp}$, the participant can decrypt the average model. Thus, we can set a small $\lambda_{sp}$ (i.e., $\lambda_{sp}=2$) to reduce the computation time. From the experiment result shown in Fig. 6(b), \textsc{CrowdFL} with a small $\lambda_{sp}$ has less computation time than PPFDL.
\begin{figure}
\centering
\subfigure[Averaging model time]{
	\includegraphics[width=0.45\linewidth]{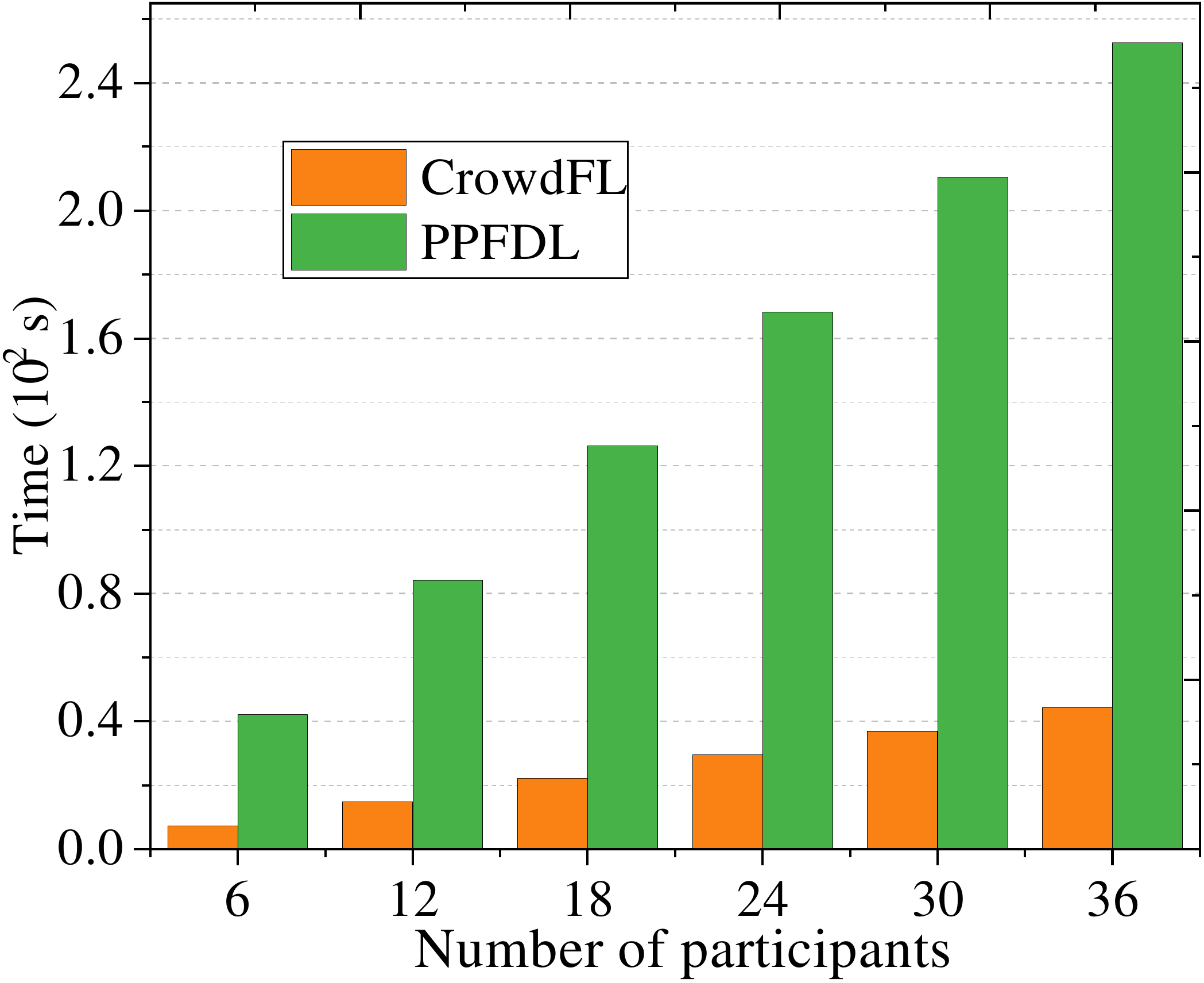}}
\hspace{0.05in}
\subfigure[Decrypting model time]{
	\includegraphics[width=0.453\linewidth]{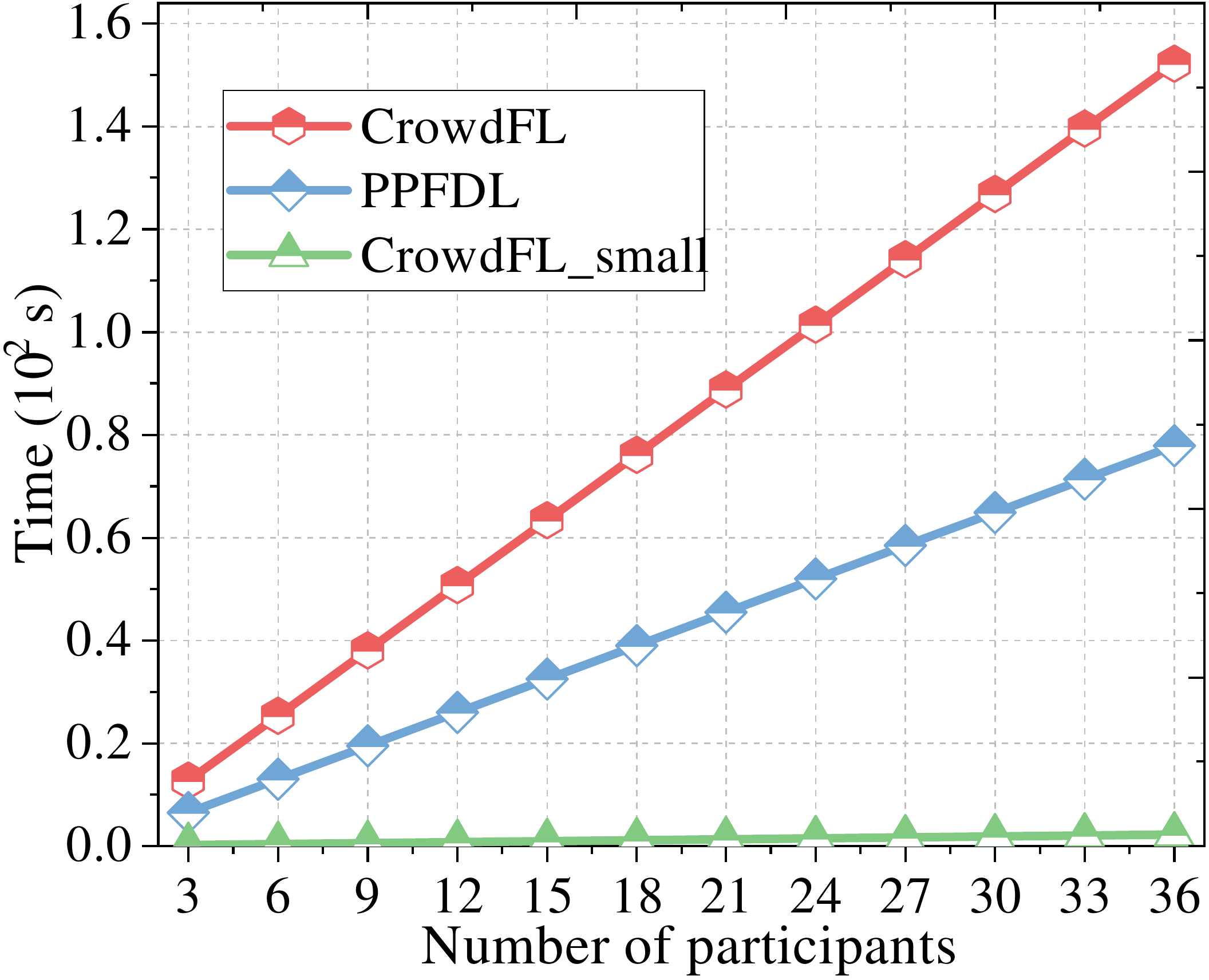}}
\caption{Comparisons of averaging model time and decrypting model time.}
\end{figure}

Fig. 7 shows the reward distribution and its running time. Specifically, ``Reward\_Rndi" means the reward distribution in $i$-th round model average, and ``Reward\_Sum" is total rewards in five rounds model aggregation. From Fig. 7(a), we can learn that if the participant can generate a more efficient training mode based on collected sensing data (i.e., the participant receives more rewards in one round model aggregation), rewards for the participant will increase round by round, whereas, rewards for the participants will reduce round by round. Fig. 7(b) gives the running time for reward distribution. Combining with Fig.6 (a) and Fig. 7(b), it can be seen that the running time of \texttt{PriFedAvg} and \texttt{PriRwd} approximately equals to that of the averaging model in PPFDL with one-time iteration. Thus, it can be concluded that \texttt{PriRwd} is efficient.
\begin{figure}
\centering
\subfigure[Reward distribution]{
	\includegraphics[width=0.45\linewidth]{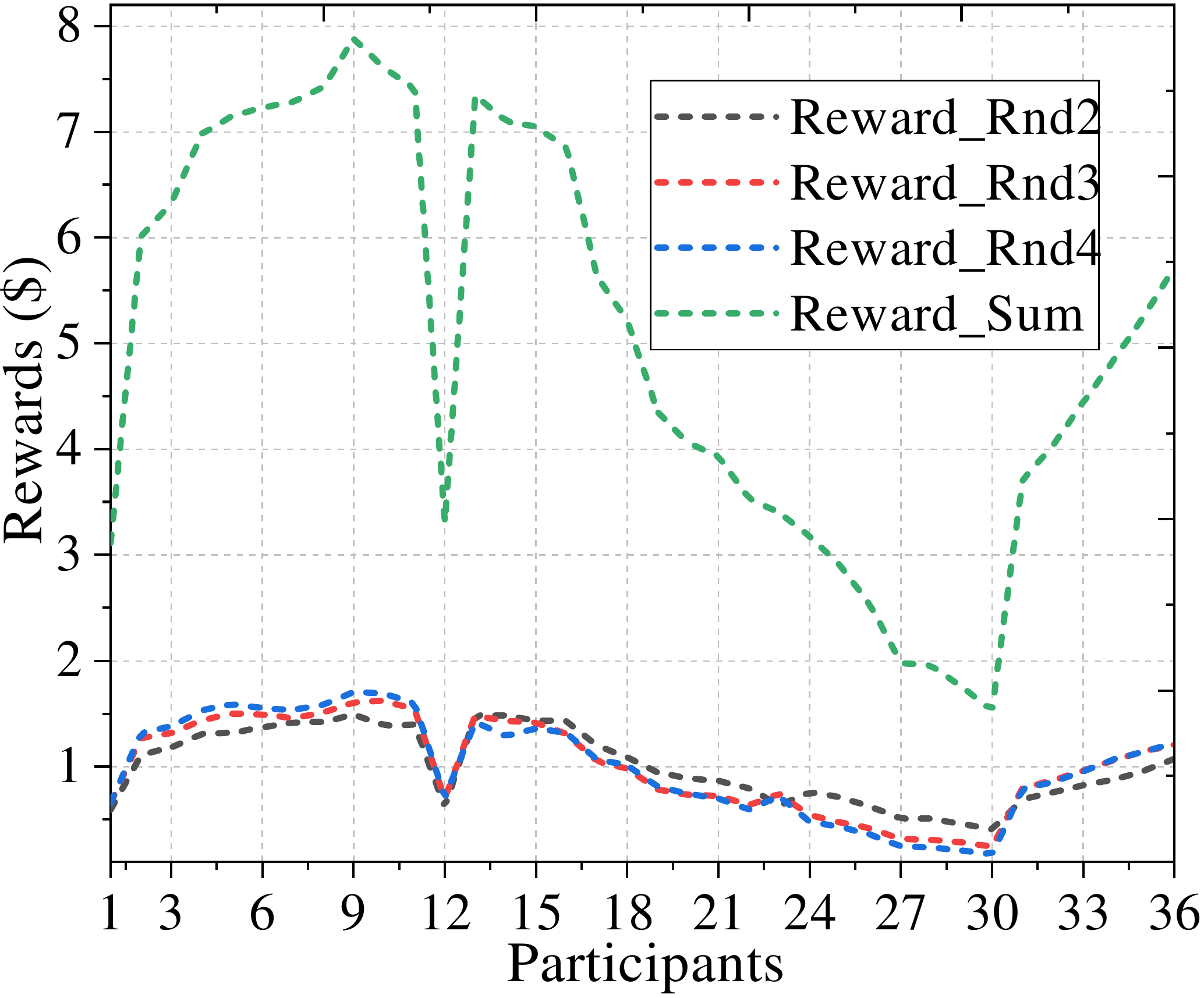}}
\hspace{0.05in}
\subfigure[Running time]{
	\includegraphics[width=0.455\linewidth]{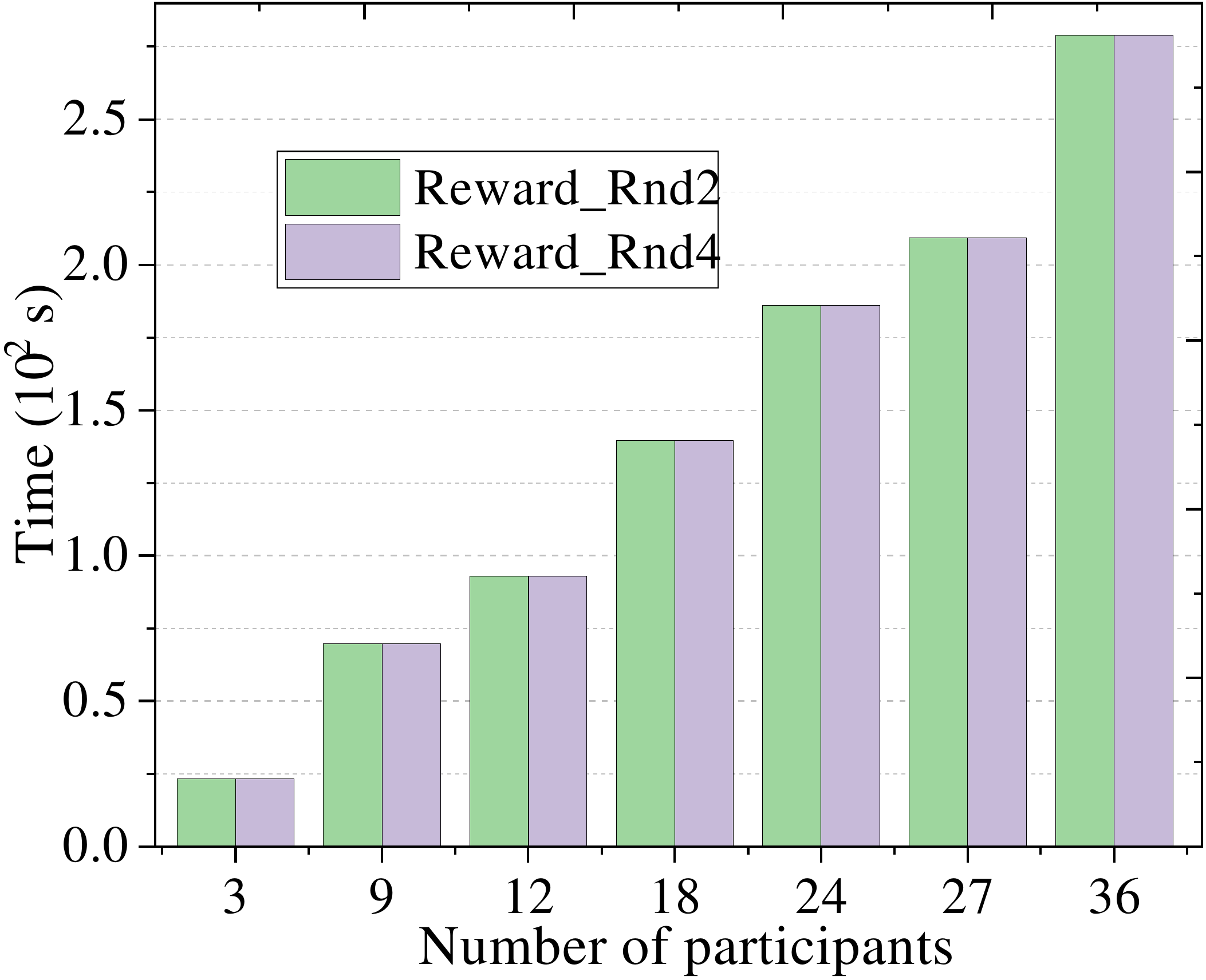}}
\caption{Reward distribution and its running time.}
\end{figure}

\section{Conclusion}
In this paper, we propose a privacy-preserving mobile crowdsensing (MCS) system based on federated learning, called \textsc{CrowdFL}. In \textsc{CrowdFL}, we outsource data collection and data processing to participants. Specifically, we propose a privacy-preserving federated averaging algorithm (\texttt{PriFedAvg}) to enable participants to locally process sensing data and protect training model privacy. Besides, we present discard and retransmission strategies to keep robustness against dropped participants and reduce the computation and communication overhead. Finally, to encourage participation, we design a privacy-preserving reward distribution mechanism (\texttt{PriRwd}) based encrypted training model. Experimental evaluations over a real-world dataset demonstrate that our proposed \textsc{CrowdFL} is feasible and efficient. In future work, we will focus on how to prevent participants from cheating in a privacy-preserving manner.

\ifCLASSOPTIONcompsoc
  \section*{Acknowledgments}
\else
  \section*{Acknowledgment}
\fi
This work was supported in part by the National Natural Science Foundation of China (Grant Nos. 62072109, U1804263, 61702105, and 61632013) and the Peng Cheng Laboratory Project of Guangdong Province PCL2018KP004.

\ifCLASSOPTIONcaptionsoff
  \newpage
\fi

\bibliographystyle{IEEEtran}
\bibliography{IEEEabrv,ref}

\begin{thebibliography}{10}
\providecommand{\url}[1]{#1}
\csname url@samestyle\endcsname
\providecommand{\newblock}{\relax}
\providecommand{\bibinfo}[2]{#2}
\providecommand{\BIBentrySTDinterwordspacing}{\spaceskip=0pt\relax}
\providecommand{\BIBentryALTinterwordstretchfactor}{4}
\providecommand{\BIBentryALTinterwordspacing}{\spaceskip=\fontdimen2\font plus
\BIBentryALTinterwordstretchfactor\fontdimen3\font minus
  \fontdimen4\font\relax}
\providecommand{\BIBforeignlanguage}[2]{{%
\expandafter\ifx\csname l@#1\endcsname\relax
\typeout{** WARNING: IEEEtran.bst: No hyphenation pattern has been}%
\typeout{** loaded for the language `#1'. Using the pattern for}%
\typeout{** the default language instead.}%
\else
\language=\csname l@#1\endcsname
\fi
#2}}
\providecommand{\BIBdecl}{\relax}
\BIBdecl

\bibitem{sucasas2020signature}
V.~Sucasas, G.~Mantas, J.~Bastos, F.~Damiao, and J.~Rodriguez, ``A signature
  scheme with unlinkable-yet-accountable pseudonymity for privacy-preserving
  crowdsensing,'' \emph{IEEE Transactions on Mobile Computing}, vol.~19, no.~4,
  pp. 752--768, 2020.

\bibitem{ryoo2017privacy}
M.~S. Ryoo, B.~Rothrock, C.~Fleming, and H.~J. Yang, ``Privacy-preserving human
  activity recognition from extreme low resolution,'' in \emph{Proceedings of
  AAAI Conference on Artificial Intelligence}, 2017, pp. 4255--4262.

\bibitem{jin2018incentive}
H.~Jin, L.~Su, H.~Xiao, and K.~Nahrstedt, ``Incentive mechanism for
  privacy-aware data aggregation in mobile crowd sensing systems,''
  \emph{IEEE/ACM Transactions on Networking}, vol.~26, no.~5, pp. 2019--2032,
  2018.

\bibitem{capponi2019survey}
A.~Capponi, C.~Fiandrino, B.~Kantarci, L.~Foschini, D.~Kliazovich, and
  P.~Bouvry, ``A survey on mobile crowdsensing systems: Challenges, solutions,
  and opportunities,'' \emph{IEEE communications surveys \& tutorials},
  vol.~21, no.~3, pp. 2419--2465, 2019.

\bibitem{liu2019boosting}
Y.~Liu, Z.~Ma, X.~Liu, S.~Ma, S.~Nepal, and R.~Deng, ``Boosting privately:
  Privacy-preserving federated extreme boosting for mobile crowdsensing,''
  \emph{arXiv preprint arXiv:1907.10218}, 2019.

\bibitem{liu2019floc}
Y.~Liu, H.~Li, J.~Xiao, and H.~Jin, ``F{Loc}: Fingerprint-based indoor
  localization system under a federated learning updating framework,'' in
  \emph{Proceedings of IEEE Conference on Mobile Ad-Hoc and Sensor Networks
  (MSN)}, 2019, pp. 113--118.

\bibitem{wang2019towards}
Z.~Wang, J.~Li, J.~Hu, J.~Ren, Z.~Li, and Y.~Li, ``Towards privacy-preserving
  incentive for mobile crowdsensing under an untrusted platform,'' in
  \emph{Proceedings of IEEE Conference on Computer Communications
  (INFOCOM)}.\hskip 1em plus 0.5em minus 0.4em\relax IEEE, 2019, pp.
  2053--2061.

\bibitem{chessa2016empowering}
S.~Chessa, A.~Corradi, L.~Foschini, and M.~Girolami, ``Empowering mobile
  crowdsensing through social and ad hoc networking,'' \emph{IEEE
  Communications Magazine}, vol.~54, no.~7, pp. 108--114, 2016.

\bibitem{zhao2020pace}
\BIBentryALTinterwordspacing
B.~Zhao, S.~Tang, X.~Liu, and X.~Zhang, ``P{ACE}: Privacy-preserving and
  quality-aware incentive mechanism for mobile crowdsensing,'' \emph{IEEE
  Transactions on Mobile Computing}, in press. [Online]. Available:
  \url{https://doi.org/10.1109/TMC.2020.2973980}
\BIBentrySTDinterwordspacing

\bibitem{mcmahan2017communication}
B.~McMahan, E.~Moore, D.~Ramage, S.~Hampson, and B.~A. y~Arcas,
  ``Communication-efficient learning of deep networks from decentralized
  data,'' in \emph{Proceedings of International Conference on Artificial
  Intelligence and Statistics (AISTATS)}, 2017, pp. 1273--1282.

\bibitem{yang2019federated}
Q.~Yang, Y.~Liu, T.~Chen, and Y.~Tong, ``Federated machine learning: Concept
  and applications,'' \emph{ACM Transactions on Intelligent Systems and
  Technology}, vol.~10, no.~2, pp. 1--19, 2019.

\bibitem{nasr2019comprehensive}
M.~Nasr, R.~Shokri, and A.~Houmansadr, ``Comprehensive privacy analysis of deep
  learning: Passive and active white-box inference attacks against centralized
  and federated learning,'' in \emph{2019 IEEE symposium on security and
  privacy (SP)}.\hskip 1em plus 0.5em minus 0.4em\relax IEEE, 2019, pp.
  739--753.

\bibitem{song2020analyzing}
M.~Song, Z.~Wang, Z.~Zhang, Y.~Song, Q.~Wang, J.~Ren, and H.~Qi, ``Analyzing
  user-level privacy attack against federated learning,'' \emph{IEEE Journal on
  Selected Areas in Communications}, vol.~38, no.~10, pp. 2430--2444, 2020.

\bibitem{yu2020fairness}
H.~Yu, Z.~Liu, Y.~Liu, T.~Chen, M.~Cong, X.~Weng, D.~Niyato, and Q.~Yang, ``A
  fairness-aware incentive scheme for federated learning,'' in
  \emph{Proceedings of the AAAI/ACM Conference on AI, Ethics, and Society
  (AIES)}, 2020, pp. 393--399.

\bibitem{bonawitz2017practical}
K.~Bonawitz, V.~Ivanov, B.~Kreuter, A.~Marcedone, H.~B. McMahan, S.~Patel,
  D.~Ramage, A.~Segal, and K.~Seth, ``Practical secure aggregation for
  privacy-preserving machine learning,'' in \emph{Proceedings of ACM SIGSAC
  Conference on Computer and Communications Security (CCS)}, 2017, pp.
  1175--1191.

\bibitem{miao2019privacy}
C.~Miao, W.~Jiang, L.~Su, Y.~Li, S.~Guo, Z.~Qin, H.~Xiao, J.~Gao, and K.~Ren,
  ``Privacy-preserving truth discovery in crowd sensing systems,'' \emph{ACM
  Transactions on Sensor Networks}, vol.~15, no.~1, pp. 1--33, 2019.

\bibitem{wang2020sparse}
L.~Wang, D.~Zhang, D.~Yang, B.~Y. Lim, X.~Han, and X.~Ma, ``Sparse mobile
  crowdsensing with differential and distortion location privacy,'' \emph{IEEE
  Transactions on Information Forensics and Security}, vol.~15, pp. 2735--2749,
  2020.

\bibitem{miao2017lightweight}
C.~Miao, L.~Su, W.~Jiang, Y.~Li, and M.~Tian, ``A lightweight
  privacy-preserving truth discovery framework for mobile crowd sensing
  systems,'' in \emph{Proceedings of IEEE Conference on Computer Communications
  (INFOCOM)}, 2017, pp. 1--9.

\bibitem{jiang2018data}
C.~Jiang, L.~Gao, L.~Duan, and J.~Huang, ``Data-centric mobile crowdsensing,''
  \emph{IEEE Transactions on Mobile Computing}, vol.~17, no.~6, pp. 1275--1288,
  2018.

\bibitem{ma2019privacy}
L.~Ma, X.~Liu, Q.~Pei, and Y.~Xiang, ``Privacy-preserving reputation management
  for edge computing enhanced mobile crowdsensing,'' \emph{IEEE Transactions on
  Services Computing}, vol.~12, no.~5, pp. 786--799, 2019.

\bibitem{aono2018privacy}
Y.~Aono, T.~Hayashi, L.~Wang, S.~Moriai \emph{et~al.}, ``Privacy-preserving
  deep learning via additively homomorphic encryption,'' \emph{IEEE
  Transactions on Information Forensics and Security}, vol.~13, no.~5, pp.
  1333--1345, 2018.

\bibitem{xu2018practical}
G.~Xu, H.~Li, and R.~Lu, ``Practical and privacy-aware truth discovery in
  mobile crowd sensing systems,'' in \emph{Proceedings of ACM SIGSAC Conference
  on Computer and Communications Security (CCS)}, 2018, p. 2312–2314.

\bibitem{xu2020privacy}
\BIBentryALTinterwordspacing
G.~Xu, H.~Li, Y.~Zhang, S.~Xu, J.~Ning, and R.~Deng, ``Privacy-preserving
  federated deep learning with irregular users,'' \emph{IEEE Transactions on
  Dependable and Secure Computing}, in press. [Online]. Available:
  \url{https://doi.org/10.1109/TDSC.2020.3005909}
\BIBentrySTDinterwordspacing

\bibitem{wen2015quality}
Y.~Wen, J.~Shi, Q.~Zhang, X.~Tian, Z.~Huang, H.~Yu, Y.~Cheng, and X.~Shen,
  ``Quality-driven auction-based incentive mechanism for mobile crowd
  sensing,'' \emph{IEEE Transactions on Vehicular Technology}, vol.~64, no.~9,
  pp. 4203--4214, 2015.

\bibitem{wang2016quality}
J.~Wang, J.~Tang, D.~Yang, E.~Wang, and G.~Xue, ``Quality-aware and
  fine-grained incentive mechanisms for mobile crowdsensing,'' in
  \emph{Processing of IEEE International Conference on Distributed Computing
  Systems (ICDCS)}.\hskip 1em plus 0.5em minus 0.4em\relax IEEE, 2016, pp.
  354--363.

\bibitem{qu2020posted}
Y.~Qu, S.~Tang, C.~Dong, P.~Li, S.~Guo, H.~Dai, and F.~Wu, ``Posted pricing for
  chance constrained robust crowdsensing,'' \emph{IEEE Transactions on Mobile
  Computing}, vol.~19, no.~1, pp. 188--199, 2020.

\bibitem{han2018quality}
K.~Han, H.~Huang, and J.~Luo, ``Quality-aware pricing for mobile
  crowdsensing,'' \emph{IEEE/ACM Transactions on Networking}, vol.~26, no.~4,
  pp. 1728--1741, 2018.

\bibitem{yang2017designing}
S.~Yang, F.~Wu, S.~Tang, X.~Gao, B.~Yang, and G.~Chen, ``On designing data
  quality-aware truth estimation and surplus sharing method for mobile
  crowdsensing,'' \emph{IEEE Journal on Selected Areas in Communications},
  vol.~35, no.~4, pp. 832--847, 2017.

\bibitem{zheng2018learning}
Y.~Zheng, H.~Duan, and C.~Wang, ``Learning the truth privately and confidently:
  Encrypted confidence-aware truth discovery in mobile crowdsensing,''
  \emph{IEEE Transactions on Information Forensics and Security}, vol.~13,
  no.~10, pp. 2475--2489, 2018.

\bibitem{liu2016efficient}
X.~Liu, K.-K.~R. Choo, R.~H. Deng, R.~Lu, and J.~Weng, ``Efficient and
  privacy-preserving outsourced calculation of rational numbers,'' \emph{IEEE
  Transactions on Dependable and Secure Computing}, vol.~15, no.~1, pp. 27--39,
  2016.

\bibitem{mohassel2017secureml}
P.~Mohassel and Y.~Zhang, ``Secureml: A system for scalable privacy-preserving
  machine learning,'' in \emph{Proceedings of IEEE Symposium on Security and
  Privacy (SP)}.\hskip 1em plus 0.5em minus 0.4em\relax IEEE, 2017, pp. 19--38.

\bibitem{xu2019verifynet}
G.~Xu, H.~Li, S.~Liu, K.~Yang, and X.~Lin, ``Verifynet: Secure and verifiable
  federated learning,'' \emph{IEEE Transactions on Information Forensics and
  Security}, vol.~15, pp. 911--926, 2020.

\bibitem{paillier1999public}
P.~Paillier, ``Public-key cryptosystems based on composite degree residuosity
  classes,'' in \emph{International conference on the theory and applications
  of cryptographic techniques}.\hskip 1em plus 0.5em minus 0.4em\relax
  Springer, 1999, pp. 223--238.

\bibitem{pei1996chinese}
D.~Pei, A.~Salomaa, and C.~Ding, \emph{Chinese remainder theorem: applications
  in computing, coding, cryptography}.\hskip 1em plus 0.5em minus 0.4em\relax
  World Scientific, 1996.

\bibitem{karaliopoulos2019optimal}
M.~Karaliopoulos, I.~Koutsopoulos, and L.~Spiliopoulos, ``Optimal user choice
  engineering in mobile crowdsensing with bounded rational users,'' in
  \emph{Proceedings of IEEE Conference on Computer Communications
  (INFOCOM)}.\hskip 1em plus 0.5em minus 0.4em\relax IEEE, 2019, pp.
  1054--1062.

\bibitem{zhao2018federated}
Y.~Zhao, M.~Li, L.~Lai, N.~Suda, D.~Civin, and V.~Chandra, ``Federated learning
  with non-{IID} data,'' \emph{arXiv preprint arXiv:1806.00582}, 2018.

\bibitem{singla2013truthful}
A.~Singla and A.~Krause, ``Truthful incentives in crowdsourcing tasks using
  regret minimization mechanisms,'' in \emph{Proceedings of International
  conference on World Wide Web}, 2013, pp. 1167--1178.

\bibitem{katz2020introduction}
J.~Katz and Y.~Lindell, \emph{Introduction to modern cryptography}.\hskip 1em
  plus 0.5em minus 0.4em\relax CRC press, 2020.

\bibitem{lyu2017privacy}
L.~Lyu, X.~He, Y.~W. Law, and M.~Palaniswami, ``Privacy-preserving
  collaborative deep learning with application to human activity recognition,''
  in \emph{Proceedings of ACM on Conference on Information and Knowledge
  Management (CIKM)}, 2017, p. 1219–1228.

\bibitem{lockhart2012applications}
J.~W. Lockhart, T.~Pulickal, and G.~M. Weiss, ``Applications of mobile activity
  recognition,'' in \emph{Proceedings of the ACM Conference on Ubiquitous
  Computing (UbiComp)}, 2012, pp. 1054--1058.

\end{thebibliography}

\end{document}